\documentclass[11pt,letter]{article}
\usepackage{multirow}
\usepackage[T1]{fontenc}
\usepackage{multicol}
\usepackage{cite}
\usepackage{booktabs}
\usepackage{mathtools}
\usepackage{xfrac}
\usepackage{xspace}
\usepackage{xargs}
\usepackage{amsmath}
\usepackage{amsthm}
\usepackage{enumerate}
\usepackage{colortbl}
\usepackage[blocks]{authblk}
\usepackage{paralist}% ipanearum
\usepackage[margin=1.0in]{geometry}
\usepackage{nccmath}
\usepackage{bbold}
\usepackage{todonotes}
\usepackage{color}
\usepackage{pbox}
\usepackage[utf8]{inputenc}
\usepackage[pdfborder={0 0 0}]{hyperref}
\usepackage{lscape}
\usepackage{pifont}
\usepackage{amsfonts}
\usepackage{amssymb}
\usepackage{times}
\usepackage{enumitem}%lets us compress the space!
\usepackage[compact]{titlesec}

\renewcommand{\checkmark}{\ding{51}}
\newcommand{\crossmark}{\ding{55}}
\newcommand{\optionalmark}{\textcolor[gray]{0.4}{(\checkmark)}}

\newtheorem{theorem}{Theorem}[section]
\newtheorem{lemma}[theorem]{Lemma}
\newtheorem{proposition}[theorem]{Proposition}

\newtheorem{definition}[theorem]{Definition}
\newtheorem{observation}[theorem]{Observation}

\newcommand{\diam}{\operatorname{diam}}

\DeclarePairedDelimiter{\ceil}{\lceil}{\rceil}
\DeclarePairedDelimiter{\floor}{\lfloor}{\rfloor}

\newcommand{\kdots}{, \dots, }

\newcommand{\etal}{{\it et~al.}}
\newcommand{\ie}{{\it i.e.}\xspace}
\newcommand{\eg}{{\it e.g.}\xspace}
\newcommand{\cf}{{\it cf.}\xspace}

\newcommand{\original}{} %could be arrow or something similar
\newcommand{\new}[1]{{#1}^+}
\newcommand{\dega}{{\ensuremath{d^+}}\xspace} %\all
\newcommand{\degs}{{\ensuremath{d^\circ}}\xspace} %self-loops
\newcommand{\dego}{{\ensuremath{d}}\xspace} %original

\newcommand{\bigo}{\mathcal{O}}

\newcommand{\arce}{{\ensuremath{\original{e}}}\xspace}

\newcommand{\avg}{{\mathbf{\bar x}}\xspace}

\newcommand{\norm}[1]{\left\lVert#1\right\rVert}

\newcommand{\out}{F^{out}}
\newcommand{\outv}[1]{\ensuremath{ {\mathbf F}_{#1}^{out}}\xspace}

\newcommand{\inc}{F^{in}}
\newcommand{\inv}[1]{\ensuremath{ {\mathbf F}_{#1}^{in}}\xspace}

\newcommand{\xv}{{\mathbf r}}

\newcommand{\del}{{\pmb {\delta}}}

\newcommand{\load}{x}
\newcommand{\loadv}{\mathbf x}
\newcommand{\flow}{f}

\newcommand{\cumflow}{F}

\newcommand{\transitionmatrix}{{P}}
\newcommand{\randomwalk}{{\mathbf P}}

\newcommand{\hatt}{\widehat{T}}

\newcommand{\RR}{{\textsc{Rotor-Router}}\xspace}
\newcommand{\RRstar}{{\textsc{Rotor-Router}}\ensuremath*\xspace}

\newcommand{\initloadv}{\loadv_1}
\newcommand{\initload}{\load_1}

\newcommand{\infnorm}[1]{\ensuremath{\norm{{#1}}_\infty}\xspace}
\newcommand{\onenorm}[1]{\ensuremath{\norm{{#1}}_1}\xspace}
\newcommand{\twonorm}[1]{\ensuremath{\norm{{#1}}_2}\xspace}

\newcommand{\eps}{{\pmb\varepsilon}\xspace}
\newcommand{\bounderror}{\ensuremath{(\delta+1)\dego}\xspace}
\newcommand{\bounderrorproof}{\ensuremath{\delta\dega+r}\xspace}

\newcommand{\tmi}{\ensuremath{t_{\mu}}\xspace}
\newcommand{\f}[1]{\ensuremath{f_{#1}}\xspace}

\newcommand{\natN}{\mathbf{N}}

\newcommand{\class}{good $s$-balancer\xspace}

\newcommand{\classpara}[1]{good $#1$-balancer\xspace}

\newcommand{\Classs}{Good $s$-balancers\xspace}
\newcommand{\classs}{good $s$-balancers\xspace}
\newcommand{\algoRF}{\textsc{Send}}

\begin{document}

\title{Improved Analysis of Deterministic Load-Balancing Schemes}

\author[1]{Petra Berenbrink}
\author[2]{Ralf Klasing}
\author[3]{Adrian Kosowski}
\author[1,4]{Frederik Mallmann-Trenn}
\author[5]{Przemys\l{}aw Uznański\footnote{Part of the work was done while the author was affiliated with LIF, CNRS and Aix-Marseille Université}}
%\author[5]{Przemysław Uznański\vspace{3mm}}

\affil[1]{Simon Fraser University, \url{ {petra, fmallman }@sfu.ca}}
\affil[2]{CNRS -- LaBRI -- Université de Bordeaux, \url{ralf.klasing@labri.fr}}
\affil[3]{Inria -- LIAFA -- Paris Diderot University, \url{adrian.kosowski@inria.fr}}
\affil[4]{École Normale Supérieure, \url{Frederik.Mallmann-Trenn@ens.fr}}
\affil[5]{HIIT -- Aalto University, \url{przemyslaw.uznanski@aalto.fi}}

\date{}
\maketitle

\begin{abstract}
We consider the problem of deterministic load balancing of tokens in the discrete model. A set of $n$ processors is connected into a $d$-regular undirected network. In every time step, each processor exchanges some of its tokens with each of its neighbors in the network. The goal is to minimize the discrepancy between the number of tokens on the most-loaded and the least-loaded processor as quickly as possible.

Rabani {\it et al.}\ (1998) present a  general technique for the
analysis of a wide class of discrete load balancing algorithms.
Their approach is to characterize the deviation between the actual loads of a discrete balancing algorithm with the distribution generated by a related
Markov chain. The Markov chain can also be regarded as the underlying model of a continuous diffusion  algorithm.
Rabani {\it et al.} showed that after time $T = \bigo(\log (Kn)/\mu)$, any algorithm of their class achieves a discrepancy  of $\bigo(d\log n/\mu)$, where $\mu$ is the spectral gap of the transition matrix of the graph, and $K$ is the initial load discrepancy in the system.

In this work we identify some natural additional conditions on deterministic balancing algorithms, resulting in a class of algorithms reaching a smaller discrepancy. This class contains well-known algorithms, \eg, the \RR.
Specifically, we introduce the notion of \emph{cumulatively fair} load-balancing algorithms
where in any interval of consecutive time steps, the total number of tokens sent out over an edge by a node is the same (up to constants) for all adjacent edges.
We prove that algorithms which are cumulatively fair and  where every node retains a sufficient part of its load in each step, achieve a discrepancy of $\bigo(\min\{d\sqrt{\log n/\mu} ,d\sqrt{n}\})$ in time $\bigo(T)$. We also show that in general neither of these assumptions may be omitted without increasing discrepancy.
We then show by a combinatorial potential reduction argument that any cumulatively fair scheme satisfying some additional assumptions achieves a discrepancy of $\bigo(d)$ almost as quickly as the continuous diffusion process. This positive result applies to some of the simplest and most natural discrete load balancing schemes.

\end{abstract}

\thispagestyle{empty}
\newpage
%\pagenumbering{arabic}

\section{Introduction}\label{intro}

In this paper, we analyze  \emph{diffusion-based load balancing algorithms}.
We assume that the processors are connected by an arbitrary $d$-regular graph $G$. At the beginning, every node has a certain number of tokens, representing its initial load. In general, diffusion-based  load balancing algorithms operate in parallel, in synchronous steps. In each step, every node balances its load with {\em all} of its neighbors.
The goal is to distribute the tokens in the network graph as evenly as possible.
More precisely, we aim at minimizing the \emph{discrepancy} (also known as \emph{smoothness}), which is defined as the difference between the maximum load and the minimum load, taken over all nodes of the network.

One distinguishes between \emph{continuous} load balancing models, in which load can be split arbitrarily, and the much more realistic \emph{discrete} model, in which load is modeled by tokens which cannot be split. In the former case, the standard continuous diffusion algorithm works as follows.
Every node $u$ having load $\load(u)$ considers all of its neighbors at the same time, and sends $\load(u) / (d+1)$ load to each of its $d$ neighbors, keeping $\load(u) / (d+1)$ load for itself. This process may also be implemented more efficiently: for each neighbor $v$ of $u$, node $u$ sends exactly $\max\{0,(\load(u)-\load(v))/(d+1)\}$ load to $v$.
It is  well-known (\cf\ \eg~\cite{SS94}) that the load in the continuous model will eventually be perfectly balanced.
In the discrete case with indivisible tokens, exact simulation of the continuous process is not possible. A node $u$ may instead try to round the amount of load sent to its neighbor up or down to an integer. Discrete balancing approaches are, in general, much harder to analyze than continuous algorithms.

In \cite{RSW98}, Rabani \etal\ suggest a framework to analyze a wide class of discrete neighborhood load balancing algorithms in regular graphs (it can be adapted to non-regular graphs).
The scheme compares the discrete balancing algorithm with its continuous version, and the difference is used to  bound the so-called error that occurs due to the rounding.
Their results hold for {\it round-fair} algorithms where any node $u$ having load $\load_t(u)$ at time $t$ sends  either
$\floor*{\load_t(u)/(d+1)}$ or $\ceil*{\load_t(u)/(d+1)}$ tokens over its edges.
They show that the discrepancy is bounded by $\bigo(d \log n / \mu)$ after $T= \bigo(\log(K n)/\mu)$ steps, where~$\mu$ is the eigenvalue gap of the transition matrix of the underlying Markov chain and $K$ is the initial load discrepancy.
The scheme applies to any discrete load balancing scheme which, at every time step, rounds the load which would be exchanged in the continuous diffusion process by a given pair of nodes to one of the nearest integers, either up or down.

The time $T$ in the above bound is also the time in which a continuous algorithm balances the system load (more or less) completely. $T$ is closely related to the {\em mixing time } of a random walk on $G$, which is the time it takes for a random walk to be on  every node with almost the same probability. Within the class of schemes considered in~\cite{RSW98}, the bound of $\bigo(d \log n / \mu)$ on discrepancy cannot be improved for many important graph classes, such as constant-degree expanders. Since the work~\cite{RSW98}, many different refinements and variants of this approach have been proposed~\cite{SS12,ABS12,FS09,FGS10}, as well as extensions to other models, including systems with non-uniform tokens~\cite{ABS12} and non-uniform machines~\cite{AB12}.

\vspace{2mm}
\subsection{Our Contribution}
\vspace{1mm}

The main goal of this paper is to continue the work of \cite{RSW98} by analyzing properties/classes of deterministic algorithms that balance better in the diffusive model than the class defined in \cite{RSW98}.
We suggest and analyze two general classes of balancing algorithms (called cumulatively fair balancers and good balancers) that include many well-known  diffusion algorithms and we bound the discrepancy they achieve after
$\bigo(T)$ time steps. Within the framework of schemes which are deterministic and pose no major implementation challenges (i.e., do not generate negative load and do not rely on additional communication), we obtain a number of significant improvements with respect to the state-of-the-art, cf.~Table~\ref{table:compare}.

For our algorithms we assume that every node of the graph has in addition to its $d$ original edges
$d^{\circ} \geq d$ many self-loops. We define $d^{+}=d^{\circ}+d$ as the degree of the graph including the self-loops.
%%%%%%%%%%%%%%%%%%%%%%%%%%%%%%%%%%%%%%%%%%%%%%%%%%%%%%%%%%%%%%%%%%%%%%%%%%%%%%%%%%

%\newcommand{\tablebgcolor}{\rowcolor[gray]{0.8}}
%\newcommand{\tablebgcolordark}{\rowcolor[gray]{.80}}
%\newcommand{\cellbgcolor}{\cellcolor[gray]{.91}}
%\newcommand{\cellbgcolordark}{\cellcolor[gray]{.80}}

\newcommand{\tablebgcolor}{}
\newcommand{\tablebgcolordark}{}
\newcommand{\cellbgcolor}{}
\newcommand{\cellbgcolordark}{}

\begin{table}[h]
\label{table:results}
\resizebox{\columnwidth}{!}{
\begin{tabular}{lllllllll}
	\toprule		
			{Algorithm} &{Discrepancy after time $\bigo(T)$} &
			{Time to reach $O(d)$ discrepancy} & {Ref.} & {D} & {SL} & {NL} & {NC} \\
			 \midrule
		{\bf Discrete diffusion with arbitrary rounding on edges} & $\bigo\left(\frac{d\log n}{\mu}\right)$ & \crossmark\ (in general) & \cite{RSW98} & \optionalmark & \optionalmark & \optionalmark & \optionalmark  \\
%&		\cite{KKM12} & $O(d \cdot n^2) $ &&& \cite{KKM12} & \checkmark & \checkmark & \checkmark & \checkmark & \crossmark \\
	Randomized distribution of extra tokens by vertices& $\bigo\left(\min\left\{d^2, d + \sqrt{\frac{d \log d}{\mu}}\right\} \sqrt{\log n}\right) $ & \crossmark& \cite{BCFFS11,SS12} & \crossmark & \checkmark & \checkmark & \checkmark \\	
%& $\bigo\left(d^2\sqrt{\log n}  \right)$ && \cite{} & &  &  &  \\
	Randomized rounding to nearest integers on edges
		 & $\bigo\left(\sqrt{d\log n} \right)$ &\crossmark& \cite{SS12} & \crossmark & \checkmark & \crossmark & \checkmark \\[2mm]
\tablebgcolor  Computation based on continuous diffusion & $\Theta(d)$ & $\bigo(T)$& \cite{ABS12} & \checkmark & \crossmark & \crossmark & \crossmark \\
\midrule
\tablebgcolor
		{\bf Cumulatively fair balancers} & $\bigo\left( \dego   \min\left\{\sqrt\frac{\log n}{\mu} ,\sqrt{n}\right\} \right)$ & \crossmark\ (in general)& Thm \ref{the:prediscrepancyinreggraphs} & \optionalmark & \optionalmark & \optionalmark & \optionalmark \\
\tablebgcolor		\phantom{o}\textbullet\ \RR &         '' & \crossmark & Thm \ref{the:prediscrepancyinreggraphs} & \checkmark & \crossmark & \checkmark & \checkmark  \\
\tablebgcolor		\phantom{o}\textbullet\ ${\bf \algoRF}\left(\floor*{\sfrac{x}{\dega}}\right)$  &         '' & open & Thm \ref{the:prediscrepancyinreggraphs} & \checkmark & \checkmark & \checkmark & \checkmark \\[2mm]

\tablebgcolordark
\cellbgcolor		\phantom{o}\textbullet\ {\bf \Classs} &         ''  & $\bigo\left(T +   \frac{\dego\log^2 n}{\mu} \right)$ & Thm \ref{the:timecprime} & \optionalmark & \optionalmark & \optionalmark & \optionalmark\\
\tablebgcolordark
\cellbgcolor		\phantom{ooo}\textbullet\ \RRstar &         ''  &         '' & Thm \ref{the:timecprime} & \checkmark & \crossmark & \checkmark & \checkmark \\
\tablebgcolordark
\cellbgcolor		\phantom{ooo}\textbullet\ ${\bf \algoRF}\left([\sfrac{x}{\dega}]\right)$, for any $ \dega > 2d $ &         ''  &       ''   & Thm \ref{the:timecprime} & \checkmark & \checkmark & \checkmark & \checkmark \\
\tablebgcolordark
\cellbgcolor		\phantom{ooo}\textbullet\ ${\bf \algoRF}\left([\sfrac{x}{\dega}]\right)$, for any $\dega \geq 3d $ &         ''  &         $\bigo\left(T +   \frac{\log^2 n}{\mu} \right)$ & Thm \ref{the:timecprime} & \checkmark & \checkmark & \checkmark & \checkmark \\
%\tablebgcolordark
%\cellbgcolor&		\phantom{ooo}\textbullet\ ${\bf \algoRF}\left(\ceil*{\sfrac{x}{3d} -\sfrac{2}{3}};\floor*{\sfrac{x}{3d}},
%\ceil*{\sfrac{x}{3d}}\right)$&         ''  & $O\left(T +   \frac{\log^2 n}{\mu} \right)$ & Thm \ref{the:timecprime} & \checkmark & \checkmark & \checkmark & \checkmark & \cellbgcolor
%		\\
[2mm]
		\bottomrule
	\end{tabular}
	}
\vspace{2mm}

{
    \small
Legend: D --- Deterministic process; SL --- Stateless process; NL --- Cannot produce negative loads;\\
\phantom{Legend: }NC --- No  additional communication required; \optionalmark\ denotes properties achieved for some implementations;\\
\phantom{Legend: }$[\cdot]$ denotes rounding to the nearest integer.
}

	\caption{
    A comparison of the discrepancy of load-balancing algorithms in the diffusive model for $d$-regular graphs. In all result statements, if not specified otherwise, we assume that the graph is augmented with at least $d$ self-loops per node.\label{table:compare}
}
\end{table}

%%%%%%%%%%%%%%%%%%%%%%%%%%%%%%%%%%%%%%%%%%%%%%%%%%%%%%%%%%%%%%%%%%%%%%%%%%%%%%%%

\paragraph{Cumulatively $\delta$-fair balancers.}
%This class can be regarded as a restriction of the round-fairness \cite{RSW98}.
We call an algorithm cumulatively $\delta$-fair if
\emph{(i)} for every interval of consecutive time steps,  the total number of tokens an algorithm sends over the original edges (non-self-loops) differs by at most a small
constant $\delta$ and
\emph{(ii)} for every time step
every edge (original edges and self-loops) of a node receives at least $\floor*{x/\dega}$ many tokens, where $x$ is the current load of the node. Cumulatively $\delta$-fair balancers are a subclass of the algorithms studied in~\cite{RSW98}, which satisfied weaker fairness conditions.
We show (Theorem \ref{the:prediscrepancyinreggraphs}) that algorithms satisfying these conditions with $d^{\circ}\ge d$ achieve a discrepancy of $\bigo(d \cdot \min\{\sqrt{\log n / \mu},\sqrt n\})$ in $\bigo(T)$ time steps.
 The restrictions result in deterministic balancing schemes with an improved discrepancy after $\bigo(T)$ time steps.
  For example, for expanders, the achieved discrepancy after time $\bigo(T)$ is $\bigo(\sqrt{\log n})$, as opposed to $\Theta(\log n)$.
%We define the class ${\bf \algoRF}\left(x\right)$, where all tokens are distributed every original edge receives $x$ tokens and every self-loop
%receives between $y$ and $z$ many tokens.

Additionally to the upper bound, we show that the discrepancy can be of order $\Omega(d \cdot \diam)$ ($\diam$ is the diameter of the graph)
if we drop the condition of cumulative fairness (Theorem \ref{lower1}) or remove
self-loops completely (Theorem \ref{lowerbound:rotorrouter}). In more detail, Theorem \ref{lower1} shows that there are round-fair balancers that satisfy the constraints of~\cite{RSW98} and that have $\Omega(d \cdot \diam)$ discrepancy nonetheless. Such discrepancy is worse than the one we obtain for cumulatively fair balancers (Theorem~\ref{the:prediscrepancyinreggraphs}) for many graph classes. In Theorem~\ref{lowerbound:rotorrouter} we show for a specific cumulatively $1$-fair balancer not using any self-loops (\RR), that it cannot achieve discrepancy better than $\Omega( n)$ on a cycle with $n$ nodes.

The class of cumulatively fair balancers contains  many well-known deterministic algorithms.
In particular, many of these algorithms are {\it stateless}, meaning the load any node sends over edges in any step depends solely on the load of the node at this time step.
 An example for a stateless cumulatively fair balancer is
 ${\bf \algoRF}\left(\floor*{\sfrac{x}{\dega}}\right)$,
where a node with load $x$ sends $\floor*{\sfrac{x}{\dega}}$ many tokens to  every original edge (non-self-loop). The remaining $x-\floor*{\sfrac{x}{\dega}}$ tokens are distributed over self-loops such that every self-loop receives at least
$\floor*{\sfrac{x}{\dega}}$ many tokens.
% ${\bf \algoRF}\left(\floor*{\sfrac{x}{\dega}};\floor*{\sfrac{x}{\dega}},\ceil*{\sfrac{x}{\dega}} + 1\right)$.
Similarly, the algorithm
${\bf \algoRF}\left([\sfrac{x}{\dega}] \right)$
%${\bf \algoRF}\left([\sfrac{x}{2d}];\floor*{\sfrac{x}{2d}},\ceil*{\sfrac{x}{2d}}\right)$
 is cumulatively fair.  Here  $[\sfrac{x}{\dega}]$ rounds $\sfrac{x}{\dega}$ to the next integer and a node
with load $x$ sends $[x/\dega]$ tokens over every original edge.
Another well-known algorithm in this class is
the so-called \RR model (also referred to in the literature as the Propp model,~\cite{AB13,FS10,CDST07}) which uses a simple round-robin approach to distribute the tokens to the $d^{+}$ neighbors, i.e., over all its edges.

%{\color{gray}
%\paragraph{Alternative version:}
%The class of cumulatively fair balancers contains  many well-known deterministic algorithms.  An example for a cumulatively fair balancer is
%${\bf \algo}\left(\floor*{x/\dega},\{ \floor*{x/\dega}, \ceil*{x/\dega} \} \right)$, where the load node $u$ sends over every original edges is
%$\floor*{x/\dega}$ and the load sent over every self-loop is in $\{ \floor*{x/\dega}, \ceil*{x/\dega} \}$.
%
%The same holds for  ${\bf \algo}\left([{x/\dega}], \floor*{x/\dega}\text{ and } \ceil*{x/\dega} \right)$, where $[y/k]$ rounds $y/k$ to the nearest integer.
%
%Another well-known algorithm in this class is
%the so-called \RR model (also referred to in the literature as the Propp model,~\cite{AB13,FS10,CDST07}) which uses a simple round-robin approach to distribute the tokens to the $d^{+}$ neighbors.
%
%}

\paragraph{\Classs.}
The class of \classs can be regarded as a restriction of cumulatively $1$-fair balancers.
A cumulatively $1$-fair balancer is a \class if it  \emph{(i)} is round-fair (every edge receives $\floor*{x/\dega}$ or $\ceil*{x/\dega}$ many tokens, where $x$ is the current load of a node) \emph{(ii)} sends over at least $s$ self-loops $\ceil*{x/\dega}$ many tokens in every round.
We show that algorithms of this class achieve a $\bigo(d)$-discrepancy within $\bigo(T + d\log^2 n / \mu)$ time steps.
Moreover, if $s=\Omega(d)$, then the same discrepancy of $\bigo(d)$ is reached within $\bigo(T + \log^2 n / \mu)$ time steps.
The class of \classs contains many well-known algorithms.
For example, the algorithm
 ${\bf \algoRF}\left([\sfrac{x}{\dega}] \right)$ is a \class for $\dega > 2d$. If $\dega \geq 3d$, then it achieves a $\bigo(d)$ discrepancy in $\bigo(T + \log^2 n / \mu)$ time.
Additionally, the class contains some variants of the rotor-router approach, \eg, one variant which we denote by \RRstar. This algorithm maintains $d-1$ self-loops, together with one special self-loop, which always receives $\ceil*{\load_{t}(u)/(2\dego)}$ tokens. The remaining tokens are distributed fairly using a rotor-router on the original edges and the $\dego - 1$ self-loops.

The only other result in the literature known to the authors which achieves a discrepancy of $\bigo(d)$ in $\bigo(T)$ steps in the {\em diffusive model} is the one of \cite{ABS12}
(see \cite{SS12} for $\bigo(d)$ discrepancy in the dimension exchange model where nodes balance with only one neighbor per round).
Their algorithms simulate and mimic (with the discrete tokens) the continuous flow.

%Another purpose was to study the relation between the discrepancy and the fraction of tokens that nodes keep for themselves (here via self-loop edges).

%however, their algorithm relies on extensive communication, computation and memory capabilities of the nodes, as well as certain additional assumptions on minimum load of nodes. (\cf~Table~\ref{table:compare} for a comparison of properties).

%The discrepancy of $O(d)$  appears to be the best possible gap achievable in reasonably short time in many graph classes (e.g., in polylogarithmic time for expanders). In fact, we provide matching lower bounds of $\Omega(d)$ on the eventual discrepancy of some of the most natural (\RR-type and \algo-type) algorithms which we consider.
%For a detailed exposition of our results, we refer the reader to Section~\ref{model} and Table~\ref{table:compare}.

%\paragraph{Algorithms in our Framework.}

%is the \algo$\left(\ceil*{\frac{x}{3d} - \frac23}\right)$ scheme, in which a node with current load $x$ sends exactly $\ceil*{\frac{x}{3d} - \frac23}$ to each of its $d$ neighbors, and keeps all the remaining tokens for itself.\NOTE{Table 1 should already be referenced at this point.}{A}

\paragraph{Technical contributions.}
%Our techniques for the analysis of cumulatively fair balancers rely on a comparison between our discrete process and of the continuous process. The latter can also be regarded as a Markovian process (random walk)  which is governed by the transition matrix of the graph.
%We calculate the total deviation (of any cumulatively fair balancer) to the continuous process as done in \cite{RSW98}. However, instead of doing it step-by-step as in \cite{RSW98}, the comparison is done over long time intervals as done in, \eg, \cite{YWB03,KP14}, in the context of the graph exploration problem.
%Our analysis connects this deviation to the value $O(\frac{\log n}{\mu})$, which is an upper bound on the mixing time.
%The authors of \cite{KP14} present a follow-up application of the approach presented here in  which they bound the cover time of the $\RR$ for specific graphs.
%For the analysis of \classs we combine the algebraic techniques used for the analysis of cumulatively fair balancers with a potential function approach.

%Main idea is to regard the  cumulative load distribution in relation to a random walk and bound the resulting 	\\
%\NOTE{Second, you should highlight the *technical* contribution of your work. Random walks, and potential functions are standard techniques for analyzing load balancing processes, and the self-preferring property is to add self-loops for every node.}{reviewer 3}
%\NOTE{ In addition to these, what are the new techniques that you use? Some discussion is needed.}{reviewer 3}
%

Our techniques for the analysis of cumulatively fair balancers rely on a comparison between our discrete process and  the continuous process. The latter can also be regarded as a Markovian process (random walk) which is governed by the transition matrix of the graph.
We calculate the total deviation (of any cumulatively fair balancer) to the continuous process as done in \cite{RSW98}. However, instead of doing it step-by-step as in \cite{RSW98}, the comparison is done over long time intervals. This was done in, \eg, \cite{YWB03,KP14}, in the context of the graph exploration problem. Our analysis of this class connects this deviation to the value $\bigo({\log n}/{\mu})$, which is a natural upper bound on the mixing time.
For the analysis of \classs we combine the algebraic techniques used for the analysis of cumulatively fair balancers with a potential function approach. Whereas all cumulatively $1$-fair balancers admit a natural potential function, which is (weakly) monotonous throughout the balancing process, for the narrower class of \classs we can define time phases so that in each time phase, the process exhibits a strict potential drop in each phase of balancing, up to a balancing discrepancy of $\bigo(d)$.
Even though we limit ourselves to regular graphs in this paper, our results can be extended to non-regular graphs.

%The remainder of this paper is structured in the following way.
%In the rest of the introduction, we set our results in the context of the state-of-the-art of the literature. In Section~2, we provide a formalization of the model, which allows us to state the required properties of cumulative fairness, strict fairness, and  {\color{red} self-preference}. We then formulate our main theorems for algorithms satisfying these properties, and provide examples of algorithms satisfying them. Sections~3 and 4 contain proofs of our main results on the discrepancy of cumulatively fair balancers and \classs  {\color{red} self-preferring} algorithms, respectively. Lower bounds on the discrepancy of some of the studied processes are presented in Section~5. We conclude with some comments on possible extensions of our results in Section~6.

%the discrepancy $\Psi$ of the considered class of diffusion schemes, and expressing it using other parameters than the spectral gap, e.g., in terms of the conductance of the graph.

\subsection{Related Work}\label{known}

Herein we only consider related results for load balancing in the discrete diffusive and balancing circuit models, and some results for the rotor-router model which are relevant to our work.

\paragraph{Diffusive load balancing.} Discrete load balancing has been studied in numerous works since~\cite{RSW98}. The authors of \cite{FGS10} propose a deterministic load balancing process in which the continuous load transferred along each edge is rounded up or down deterministically, such that the sum of the rounding errors on each edge up to an arbitrary step $t$ is bounded by a constant. This property  is called the {\em  bounded-error property}. Then they show  that after $T$ steps any process with bounded-error property achieves a discrepancy  of $\bigo(\log^{3/2} n)$ for hypercubes and $\bigo(1)$  for constant-degree tori. There are no similar results for other graph classes. Note that the algorithm of \cite{FGS10} has the problem that  the original demand of a node might exceed its available load,  leading to so-called {\em negative load}.

In \cite{AB13}, the authors consider \RR-type walks as a model for load balancing. It is assumed that half of the edges of every node are self-loops. The authors present an algorithm which falls in the class of  {\em bounded-error diffusion processes} introduced in \cite{FGS10}.
This results in  discrepancy bounds of~$\bigo(\log^{3/2} n)$ and~$\bigo(1)$ for hypercube and $r$-dimensional torus with $r=\bigo(1)$. In \cite{AB12}, the authors consider the diffusion algorithms that always round down for heterogeneous  networks. They also  show that a better load balance can be obtained when the algorithm is allowed to run longer than $T$ steps.

In \cite{ABS12}, the authors propose an algorithm that achieves discrepancy of $2 d$ after $T$ steps for any graph. For every edge $e$ and step $t$, their algorithm calculates the number of tokens that should be sent over $e$ in $t$ such that the total number of tokens forwarded over $e$ (over the first $t$ steps) stays as close as possible to the amount of load that is sent by the continuous algorithm over $e$ during the first $t$ steps. However, their algorithm can  result in negative load when the initial load of any node is not sufficiently large and it has to calculate the number of tokens that the continuous algorithm sends over all edges.
Note that the algorithm presented in this paper has to simulate the continuous algorithm in order to calculate the load that it has to transfer over any edge, whereas our algorithms are much easier and they do not need any additional information, not even the load of their neighbors.

There are several publications that suggest randomized rounding schemes\cite{ABS12,BCFFS11, FGS10, AB12b, AB13, SS12} to convert the continuous load that is transferred over an edge  into discrete load.
%The randomized algorithm of~\cite{FGS10} achieves an imbalance bound of~$O(d\log \log n /\mu)$ for $d$-regular graphs,~$O(d\log \log n)$ for expanders,~$O(\log^2 n\log \log n)$ for hypercubes, and~$O(n^{1/2d}\log \log n)$ for tori.
The algorithm of~\cite{BCFFS11} calculates the number of {\em additional} tokens (the difference between the continuous flow forwarded over edges and the number of tokens forwarded by the discrete algorithm after rounding {\em down}). All additional tokens  are sent to randomly chosen neighbors.
Their discrepancy bounds after $T$ steps are~$\bigo(d\log \log n /\mu)$ and~$\bigo(d\sqrt{\log n}+\sqrt{d\log n\log d/\mu})$ for $d$-regular graphs,~$\bigo(d\log \log n)$ for expanders,~$\bigo(\log n)$ for hypercubes, and~$\bigo(\sqrt{ \log n})$ for tori.
%In~\cite{ABS12}, the authors also consider a randomized version of their protocol for arbitrary graphs. The discrepancy bound is w.h.p. ~$O(\sqrt{d\log n})$ after $T$ time steps if the initial load on every node is at least $d/4+c\sqrt{d\log n}$ for some constant $c>0$.
The authors of \cite{SS12} present two randomized algorithms for the diffusive model. They achieve after $\bigo(T)$ time a discrepancy of $\bigo(d^2\sqrt{\log n})$ by first sending $\floor*{\load(u) / (d+1)}$ tokens to every neighbor and itself, and afterwards by distributing the remaining tokens randomly. Additionally, they provide an algorithm which achieves after $\bigo(T)$ time a discrepancy of $\bigo(\sqrt{d\log n})$ by rounding the flow sent over edges randomly to the nearest integers which might cause negative loads.
For a comparison with our results see Table \ref{table:compare}.

\paragraph{Dimension exchange model.}
In the Dimension Exchange model, the nodes are only allowed to balance with one neighbor at a time.  Whereas for all diffusion algorithms considered so far the discrepancy in the diffusion model is at least $d$, dimension exchange algorithms are able to balance the load up to an additive constant. In \cite{FS09} the authors consider a discrete dimension exchange algorithm for the matching model. Every node~$i$ that is connected to a matching edge calculates the load difference over that edge. If that value is positive, the algorithm rounds it up or down, each with probability one half. This result is improved in \cite{SS12}, where the authors show that a constant final discrepancy can be achieved within~$\bigo(T)$ steps for regular graphs in the random matching model, and constant-degree regular graphs in the periodic matching (balancing circuit) model.

\paragraph{Rotor-router walks.}

Originally introduced in~\cite{PDDK96}, the \emph{rotor-router walk} model was  employed by Jim Propp for derandomizing the random walk, thereby frequently appearing under the alternative names of {\em Propp machines} and {\em deterministic random walks}~\cite{CDST07,DF09,FS10,KKM12}.
In the rotor-router model, the nodes send their tokens out in a round-robin fashion. It is assumed that the edges of the nodes are cyclically ordered, and that every node is equipped with a rotor which points to one of its edges. Every node first sends one token over the edge pointed to by the rotor. The rotor is moved to the next edge which will be used by the next token, and so on, until all tokens of the node have been sent
out over one of the edges. It has been shown that the rotor walks capture the average behaviour of random walks in a variety of respects such as  hitting probabilities and hitting times.
The rotor-router model can be used for load balancing, and directly fits into the framework we consider in this paper.
The authors of \cite{journals/corr/ShiragaYKY13} obtain rough bounds on the discrepancy, independently of \cite{RSW98}.
 In \cite{AB13}, the authors study a lazy version of the rotor-router process (half of the edges are self-loops) for load balancing. They prove that the rotor walk falls in the class of  {\em bounded-error diffusion processes} introduced in \cite{FGS10}. Using this fact they obtain discrepancy bounds of~$\bigo(\log^{3/2} n)$ and~$\bigo(1)$ for the hypercube and $r$-dimensional torus with $r=\bigo(1)$, respectively, which improve the best existing bounds of~$\bigo(\log^2 n)$ and~$\bigo(n^{1/r})$ in this graph class.

\subsection{Model and Notation}
\label{sec:mod}

In this section we define our general model, which applies to both classes of studied algorithms.

The input of the load-balancing process is a symmetric and directed regular graph $G=(V,\original{E})$ with $n$ nodes. Every node has out-degree and in-degree $\dego$. We have $m\in\natN$ indivisible tokens (workload) which are arbitrarily distributed over the nodes of the network. For simplicity of notation only, we assume that initially $G$ does not contain multiple edges. In general, the nodes of the graph may be treated as anonymous, and no node identifiers will be required.
The time is divided into synchronized steps. Let $\loadv_t=(\load_t(1)\kdots \load_t(n))$ be the {\it load vector} at the beginning of step $t$, where $\load_t(u)$ corresponds to the load of node $u$ at the beginning of step $t$. In particular, $\initloadv$ denotes the initial load distribution and $\avg$ is the real-valued vector resulting when every node has achieved average load, $\avg(u) = \frac1n \sum_{u\in V}\load_1(u)  \equiv \bar\load$, for all $u\in V$. Note that the total load summed over all nodes does not change over time.
The {\it discrepancy} is defined as
the load difference between the node with the \emph{highest load} and the node with the \emph{lowest load}.
We will denote by $K$ the maximal initial discrepancy in $\initloadv$, \ie, $K=\max_{u\in V} \initload(u) -\min_{u\in V} \initload(u)$.
 The {\it balancedness} of an algorithm is defined as the gap between the node with the highest load and the \emph{average load}.

In order to introduce self-loops, we transform $G$ into the graph $\new{G}=(V,\original{E} \cup E^\circ)$ by adding  $\degs$ self-loops to every node. For a fixed node $u\in V$, let $ E^\circ_u=\{e_1(u,u)  \kdots e_{\degs}(u,u) \}$ denote the set of self-loops of  $u$ and let $\original{E}_u$ denote the original edges of $u$ in $G$. We assume $\degs = \bigo(\dego)$. We can now define $\new{E}_u = E^\circ_u \cup \original{E}_u \mbox{ \ and \ } E^\circ=\bigcup_{u\in V}  E^\circ_u.$ In the following, we call $G$ the {\em original graph} and $\new{G}$ the {\em balancing graph}. The edges $E^\circ$  are called self-loop edges  of
$\new{G}$ and $\original{E}_u$ are the original edges. We remark again that the balancing graph is introduced for purposes of analysis, only, and is completely transparent from the perspective of algorithm design.
We also define $\dega=\dego + \degs$ as the degree of any node in $\new{G}$. $N(u)$ is the set of direct neighbors of $u$ in $\new{G}$, \ie, it contains all neighbors of $u$ in $G$ including $u$ itself (because of the self-loops).

For a fixed edge  $e=(u,v)\in \new{E}$ let $\f{t}(e) $ be the number of tokens which $u$ sends to $v$ in step  $t$. In particular, let $f_t(u,u) = \sum_{e\in E_u^\circ}f_t(e)$.
Let $F_t(e)$ denote the cumulative load sent from $u$ to $v$ in steps $1 \kdots t$, i.e,
$\cumflow_t(e)=\sum_{\tau \leq t} \flow_\tau(e).$
For a fixed node $u$, we define $f^{out}_t(u)=\sum_{v\in N(u)} f_t(u,v)$  to be the number of tokens (or flow) leaving $u$ in step $t$.
The incoming flow is then defined as $f^{in}_t(u)=\sum_{\{v : u\in N(v)\}} f_t(v,u).$
We define the cumulative incoming and original flows $\out_t(u)$ and $\inc_t(u)$ accordingly, and
$\outv{t}$ and $\inv{t}$  are defined as the vectors of the (cumulative) original and incoming flow.
We will say that  edge $e=(u,v)$ received $x$ tokens when node $u$ sends $x$ tokens to  $v$.

\section{Results for Cumulatively Fair Balancers}\label{s:fair}

In this section we present a general class of algorithms, called {\em cumulative fair} algorithms, and analyze their discrepancy after
 $T=\bigo\left((\log K + \log n)/\mu\right)$ many time steps. Note that $T$ is the balancing time of the continuous diffusion algorithm if the initial discrepancy is $K$,  and $\dego$ is the number of original edges (non-self-loops) of each node.

 % In Section \ref{general} we define cumulatively fair balancers in detail and present our results. %In Section \ref{analysis1} we finally analyze algorithms of this class.

 %\subsection{General Model and Results}\label{general}

%%%%%%%%%%%%%%%%%%%%%%%%%%%%%%%%%%%%%%%%%%%%%%%%%%%%%%%%%%%%%%%%%%%%%%%%%%%%%%%%%%%%%%%%%%%%%%%%%%%%%%%

We call an algorithm {\em cumulatively fair} if the  flow that is sent out over every edge of $u$ (including the self-loops) up to step $t$ can differ by at most  $\delta$.

\begin{definition}\label{fair}
Let $\delta$ be a constant.
An algorithm is called  {\em cumulatively} $\delta$-{\em fair}  if for all
$t\in\natN$, $u\in V$
\begin{itemize}
\item every edge  $\original{e}\in E_u^+$ receives at least $\floor*{x_t(u)/\dega} $ many tokens.
\item all original edges  $\original{e_1}, \original{e_2}\in E_u$ satisfy $|\cumflow_{t}(\original{e_1}) - \cumflow_{t}(\original{e_2}) |\leq \delta$.
\end{itemize}
\end{definition}
Note that {\em round-fair} algorithms (defined in \cite{RSW98}) are not necessarily
cumulatively $\delta$-fair for any fixed $\delta$.
\begin{observation}
The algorithms
 ${\bf \algoRF}\left(\floor*{\sfrac{x}{\dega}}\right)$ and
 ${\bf \algoRF}\left([{\sfrac{x}{\dega}]}\right)$
are $0$-cumulatively fair.
Furthermore, \RR is cumulatively $1$-fair.
\end{observation}
For cumulatively fair balancers we show the following bound.
\begin{theorem}\label{the:prediscrepancyinreggraphs}
Let $G$ be any $d$-regular input graph and let $d^+$ is the degree of the balancing graph $G^{+}$ (including the self-loops).
Let $\delta$ be a constant.
Assume $A$ is a cumulatively $\delta$-fair balancer.
Then, after $\bigo\left(\log(K n)/\mu\right)$ time steps, $A$ achieves a discrepancy of
\begin{enumerate}[label=(\roman*),topsep=0pt, partopsep=0pt]%[(i)]
 \setlength{\itemsep}{1pt}%
    \setlength{\parskip}{0pt}
\item \label{en1:part1} $\bigo\left( (\delta+1)\cdot \dego \cdot
\sqrt{\log n/\mu} \right)$ for $\dega \ge 2\dego$.
\item \label{en1:part2}  $\bigo\left((\delta+1)\cdot\dego\cdot \sqrt{n}\right)$ for $\dega\ge 2\dego$.
\item \label{en1:part3}  $\bigo\left((\delta+1) \cdot\dego
\cdot{\log n/\mu} \right)$ for arbitrary $\dega\ge \dego + 1$.
\end{enumerate}

  \end{theorem}

For constant $\delta$ and at least $d$ self-loops, the results of Claim \ref{en1:part1} of the above theorem show a better discrepancy after $T$ steps compared to the result of \cite{RSW98}. %For a smaller number of self-loops, we were unable to show the same discrepancy without increasing the required time.
Claim \ref{en1:part2} provides an improvement for graphs with a bad expansion (small eigenvalue gap), such as cycles.

%The proof of the theorem can be found in Section~\ref{analysis1}.

%\subsection{Analysis of the Class}\label{analysis1}

%In this section, we analyze the broad class of cumulatively fair load balancing algorithms.

The rest of this section is devoted to the proof of Theorem~\ref{the:prediscrepancyinreggraphs}. The core idea of the proof is to regard the balancing process over several steps and to observe that the cumulative load of the nodes  is closely related to a random walk of all tokens.
The difference to the random walk can be bounded by a small corrective vector  depending on $\delta$.
%Before continuing with the proof, we require some further definitions which we provide in the next section.
We start by providing some additional definitions.

%\paragraph{Preliminaries for the proof.}
In the analysis, we change the way the tokens are kept at a node:
In addition to sending tokens over self-loop edges we allow a node to retain a remainder (of tokens) of size $r<\dega$ at every node. The reason for this is that the proof requires the cumulative fairness on all edges and not just on original edges.
One can show (Proposition \ref{pro:SimulateCumFairness}) that every cumulatively fair balancer can be transformed (by shifting tokens from self-loops to the remainder) into an algorithm guaranteeing cumulative fairness on all edges and that this transformed algorithm sends exactly the same load over original edges in every round.%
%For the analysis, we assume that cumulatively fair balancers have two different means to retain load and not to send their neighbors:  \emph{(i)} so-called {\em remainders} and \emph{(ii)} the self-loops defined above in Section \ref{general}. The remainders facilitate the analysis of the algorithms in the next section.

Let the remainder  $r_t(u)$ of node $u$ in step $t$  be the number of tokens of $u$ that will not participate in the load distribution over its original edges and self-loops. Then, $\xv_t=(r_t(1) \kdots r_t(n))$ denotes the {\it remainder vector} at step $t$, where $r_t(i)$ is the number of tokens kept by the $i$'th node of $G$ in step $t$. We will denote by $r$ the upper bound on the maximum remainder of an algorithm, satisfying $|r_t(u)|\le r \leq \dega$ for every time step $t$ and every $u\in V$. For the ease of notation, we assume that $r<d^{+}$.
Note that for all $u\in V$ and all steps $t$
\begin{equation}\label{eqn:inout} \initload(u) + \inc_{t-1}(u)=r_t(u) + \out_t(u).
\end{equation}
Moreover,
\begin{equation}\label{eqn:flowout} \load_t(u) = f_{t}^{out}(u)+ r_t(u).
\end{equation}

%\paragraph{Random walks.}
Let $\randomwalk$ denote the transition matrix of a random walk of $G^+$. Let
$\randomwalk(u,v)$ denote the one-step probability for the walk
to go from $u$ to $v$. Then  $\randomwalk(u,v)= 1/\dega$ if
$(u,v)\in \original{E}$, $\randomwalk(u,v)=\degs/\dega$ if  $u=v$, and $\randomwalk(u,v)= 0$ otherwise.

%\[
%     \randomwalk(u,v)=\left\{\begin{array}{ll} 1/\dega & (u,v)\in \original{E} \\
%\degs/\dega & u=v \\
%         0 & otherwise\end{array}\right.. %one dot is eaten by right
%  \]
% \\
Let $\mu$ be the eigenvalue gap of $\mathbf{P}$, i.e., $
\mu = 1-\lambda_2$, where $\lambda_2$ is the second largest eigenvalue.
We define $\randomwalk^t$ to be the $t$-steps transition matrix, \ie, $\randomwalk^t=
\randomwalk\cdot \randomwalk^{t-1}$. We define the {\em steady-state} distribution as
$\randomwalk^\infty=\lim_{t\rightarrow \infty} \randomwalk^t$.
Note that $\forall u,v \in V$, $ \randomwalk^\infty(u,v)=\dega/(2|\new{E}|)=1/n.$
Observe that $ \transitionmatrix^\infty \cdot \initloadv=(\bar\load,\bar\load,\ldots, \bar\load).$
We can express $\randomwalk^t$ as $\randomwalk^t = \randomwalk^\infty + {\bf \Lambda} _t$, where $ {\bf \Lambda} _t$ is the error-term calculating the difference between $\randomwalk^t$
and the steady-state distribution.
%We define similar to \cite{LevinPeresWilmer2006}
%\vspace{-0.3cm}
% $$\mbox{diff}(t) = \max_{\onenorm{{\bf q} }=1}\{ \onenorm{ (\randomwalk^t - \randomwalk^\infty){\bf q}  }  \}  = \max_{\onenorm{{\bf q} }=1}\{ \onenorm{{\bf \Lambda} _t {\bf q} }  \}. $$
%The mixing time $t_{mix}$ is defined as $
%\minarg_t \{ \mbox{diff}(t) \leq 1/2 \}.$
%\paragraph{Vector norm.}
Let $p \geq 1$.
The $p$-norm of a vector ${\mathbf r}$ is defined as
$\|{\mathbf r} \|_p=\left(\sum_{i=1}^n |r_i|^p \right)^{\frac{1}{p
}}.$
In particular, $\infnorm{{\mathbf r}}$ is defined to be $\max \{ |r_1| \kdots |r_n|\}$.

\medskip
We are ready to prove the main theorem of this section.
The core idea of the proof is to calculate the total deviation between any cumulatively fair balancer and a continuous process, similar to \cite{RSW98}. However, instead of comparing the two processes step-by-step as in \cite{RSW98}, the comparison is done over long time intervals similar to \cite{KP14}. This deviation is then connected to the value $\bigo({\log n}/{\mu})$.

%
%\subsubsection*{Proof of Theorem \ref{the:prediscrepancyinreggraphs} Part 1.}\label{s:proofofthm1}
\begin{proof}[Proof of Theorem \ref{the:prediscrepancyinreggraphs}]
Fix a node $u\in V$.
Note that by the definition of cumulatively $\delta$-fairness
and Proposition \ref{pro:SimulateCumFairness}
 we have for all
$(u,v) \in {E}_u^+$ that\
\begin{equation}\label{eqn:diffdelta2}
\hspace{1.5cm} \left| \cumflow_t(u,v) - {\out_t(u)}/{\dega} \right|\leq \delta.
\end{equation}
We will define a corrective vector which at time $t$ measures the difference between
the load the nodes sent over original edges at time $t$ and
the load the nodes should have sent over these edges in order to
ensure that every original edge received the exact same (continuous) load until time $t$.
Formally, we define  the $n$-dimensional  corrective vector $\del_{t,u}$ with $\delta_{t,u}(v)=F_t(v,u)-\out_t(v)/\dega$ for $v\neq u$ and
$\delta_{t,u}(u)=F_t(u,u)-\degs/\dega\cdot \out_t(u)$.
The entries of the corrective vector satisfy $|\delta_{t,u}(v)| \leq \delta $ for $v \in N(u) \setminus \{ u \}$, $|\delta_{t,u}(u)| \leq \degs\delta $, and $\delta_{t,u}(v)=0$ for $v \not\in N(u)$.
 Consequently,
 $\onenorm{\del_{t,u}}\leq \delta \dega$.
Then
we derive from (\ref{eqn:diffdelta2}) the following bound on the incoming cumulative load of node $u$
\begin{align} \inc_t(u)&\underset{def.}{=}\sum_{v \in N(u) }\cumflow_t(v,u)\notag
= \sum_{v \in N(u)\setminus\{u\} }\cumflow_t(v,u) + \cumflow_t(u,u)\notag\\
&\underset{ (\ref{eqn:diffdelta2})}{=}
\sum_{v\in N(u)\setminus\{u\} }\left(\frac{1}{\dega}\out_{t}(v)+ \delta_{t,u}(v)  \right)+ \frac{\degs}{\dega} \out_{t}(u)+\delta_{t,u}(u)\notag\\
 &= \sum_{v\in N(u)\setminus\{u\}}\frac{1}{\dega}\out_{t}(v) + \frac{\degs}{\dega}\out_{t}(u)  + \onenorm{\del_{t,u}} . \label{eqn:rewrtingin}
\end{align}

Rewriting (\ref{eqn:inout}) by introducing (\ref{eqn:rewrtingin}) we get:
\begin{equation}\label{eqn:rewrtinginvec}\out_t(u)=\initload(u)+ \sum_{v\in N(u)\setminus\{u\}}\frac{1}{\dega}\out_{t-1}(v) + \frac{\degs}{\dega}\out_{t-1}(u)+ \underbrace{\onenorm{\del_{t,u}}- r_t(u)}_{\varepsilon_t(u)}. \end{equation}

We have $\infnorm{\varepsilon_t(u)} \leq \bounderrorproof$.
Rewriting (\ref{eqn:rewrtinginvec}) in vector form, we obtain
$$\outv{t} = \initloadv + \randomwalk \cdot \outv{t-1} + \eps_t = \sum_{0 \leq \tau < t} \randomwalk^\tau\initloadv + \sum_{1 \leq \tau \le t} \randomwalk^{t-\tau} \cdot \eps_\tau.$$
%Expanding with respect to $t$ gives
%$\outv{t}= \sum_{0 \leq \tau < t} \randomwalk^\tau\initloadv + \sum_{1 \leq \tau \le t} \randomwalk^{t-\tau} \cdot \eps_\tau.$

For any $\hatt>0$, the number of tokens leaving node $u$  in the interval  $[t+1; t+\hatt]$ is $\outv{t+\hatt}(u)-\outv{t}(u)$. Hence,
 $$\outv{t+\hatt}-\outv{t}=\sum_{t \le \tau < t+\hatt}    \randomwalk^\tau \initloadv +
\sum_{1\leq \tau\leq t}(\randomwalk^{t+\hatt-\tau}-\randomwalk^{t-\tau})\cdot \eps_\tau +
\sum_{  t < \tau \le t+\hatt}\randomwalk^{t+\hatt-\tau}\cdot\eps_\tau.
$$
We set $t^*=t-4\tmi$, where $\tmi = 6\log n / \mu$, and substitute $\randomwalk^\tau = \randomwalk^\infty + {\bf \Lambda} _\tau$. We derive
\begin{align}
&\outv{t+\hatt}-\outv{t}=
\sum_{\mathclap{t \le \tau < t+\hatt}}  \randomwalk^\tau \initloadv +
\sum_{\mathclap{1\leq \tau\leq t^*}}(\randomwalk^{t+\hatt-\tau}-\randomwalk^{t-\tau})\eps_\tau +
\sum_{\mathclap{t^* < \tau\leq t}}(\randomwalk^{t+\hatt-\tau}-\randomwalk^{t-\tau})\eps_\tau +
\sum_{ \mathclap{ t < \tau \le t+\hatt}}\randomwalk^{t+\hatt-\tau}\eps_\tau \notag\\
&=
\hatt\randomwalk^\infty\initloadv
+ \sum_{\mathclap{t\le \tau< t+\hatt}}{\bf \Lambda} _\tau\initloadv \phantom{o}+
\sum_{\mathclap{1 \le \tau\leq t^*}}({\bf \Lambda}_{t+\hatt-\tau} - {\bf \Lambda} _{t-\tau})\eps_\tau\!+\!\sum_{\mathclap{t^* < \tau \leq t}}(\randomwalk^{t+\hatt-\tau} - \randomwalk^{t-\tau})\cdot \eps_\tau\!+\!\sum_{\mathclap{t <\tau \leq t+\hatt}}\randomwalk^{t+\hatt-\tau}\cdot \eps_\tau.
\label{diffinout}
\end{align}
%
%Therefore, the maximum difference taken over all nodes between the average load of the node in the last $\hatt$ time steps and the average load $\bar\load$, after multiplying by $\hatt$ is
In the following we use $ {\bf\bar\load} = \randomwalk^\infty\initloadv $.
\begin{align*}
\infnorm{ \sum_{t < \tau \leq t+\hatt }\loadv_\tau - \hatt\cdot {\bf\bar\load}}  &\underset{(\ref{eqn:flowout})}{=}  \infnorm{\left(\outv{t+\hatt} - \outv{t} +
\sum_{t < \tau \leq t+\hatt }\xv_\tau\right) - \hatt \randomwalk^\infty\initloadv}\\
&\underset{(\ref{diffinout})}{\leq}  \sum_{t \le \tau< t+\hatt}\infnorm{{\bf \Lambda}_\tau\initloadv}
\phantom{o} +\sum_{\mathclap{1\le\tau\leq t^*}}\left(\infnorm{{\bf \Lambda} _{t+\hatt-\tau} \eps_\tau}
+\infnorm{ {\bf \Lambda} _{t-\tau}\eps_\tau} \right)\\
&\phantom{oo}+\sum_{\mathclap{t^* < \tau \leq t}}\infnorm{(\randomwalk^{t+\hatt-\tau} - \randomwalk^{t-\tau})\eps_\tau }
\phantom{o}+\sum_{\mathclap{t <\tau \leq t+\hatt}}\infnorm{\randomwalk^{t+\hatt-\tau} \eps_\tau} \phantom{o}+ \sum_{\mathclap{t < \tau \leq t+\hatt} }\infnorm{\xv_\tau}.  %%%%
\end{align*}
By well-known properties of mixing in graphs (cf.\ claims (i) and (ii) of Lemma \ref{lem:boundmix} in the Appendix) it follows for $t \geq 16 \cdot\log(nK)/\mu$ that $\forall_{\tau \ge t} \infnorm{{\bf \Lambda} _\tau\initloadv} \le 2^{-4}$, and  moreover $\sum_{\tau \geq 4 \cdot  \tmi }\infnorm{{\bf \Lambda} _\tau \eps_\tau} \leq  n^{-4} \cdot \max_{\tau\geq 4 \cdot \tmi}\{\infnorm{\eps_\tau} \}.$
This results in
\begin{align*}
\infnorm{ \sum_{t < \tau \leq t+\hatt }\loadv_\tau - \hatt\cdot {\bf\bar\load}}
&\leq \hatt \cdot 2^{-4} \phantom{o}+2 \cdot n^{-4}  \max_{1\le\tau\leq t^* } \infnorm{\eps_\tau}
+\sum_{t^* < \tau \leq t}\infnorm{(\randomwalk^{t+\hatt-\tau} - \randomwalk^{t-\tau})\eps_\tau } \\
&\phantom{oo}+\hatt(\bounderrorproof)
\phantom{o}+ \hatt r  \\
&\leq \frac{\hatt}{4} + (\bounderrorproof)  +\sum_{t^* < \tau \leq t}\infnorm{(\randomwalk^{t+\hatt-\tau} - \randomwalk^{t-\tau})\eps_\tau }  +  \hatt(\delta\dega + 2r).
\end{align*}
Dividing the last equation by $\hatt$ yields
 \begin{equation}\label{eqn:theT}\infnorm{ \frac{\sum_{t < \tau \leq t+\hatt }\loadv_\tau}{\hatt} -  {\bf \bar\load}} \leq \frac{1}{4} +  (\delta\dega + 2r)+\frac{ (\bounderrorproof) +  \sum_{t^* < \tau \leq t}\infnorm{(\randomwalk^{t+\hatt-\tau} - \randomwalk^{t-\tau})\eps_\tau }}{\hatt}.
\end{equation}
%
%
%
%
%begin merge
%
In this way, in (\ref{eqn:theT}) we have derived a bound on the difference between the average load of a node during an interval of length $\hatt$ and the average load $\bar x$.

In the remainder we bound (\ref{eqn:theT}) for $\hatt =1$, which describes the load difference to the average. % We will relate it to the mixing time and we will use Lemma \ref{lem:boundmix} to show that the corrective vectors become diminishingly small.
Fix an arbitrary $t \geq 16 \log(nK)/\mu$.
%Let $\bf w^{\top}$ denote the transposed vector.
We define $P_t(u,w)$ to be the probability that a random walk following matrix $\randomwalk$, initially located at $u\in V$, is located at $w$ after $t$ time steps. We then obtain the following bound (see Appendix~\ref{sec:appendix23} for details):
\begin{equation}
 \infnorm{(\randomwalk^{t+1-\tau} - \randomwalk^{t-\tau}) \eps_\tau }
\leq (\bounderrorproof)  \cdot \max_{{ w}\in V} \sum_{v\in V} \left| P_{t+1-\tau}(v,w)-P_{t-\tau}(v,w) \right| . \label{eqn:boundtequals1}
\end{equation}
Combining \eqref{eqn:theT} and \eqref{eqn:boundtequals1} for $\hatt=1$, recalling that $t^* = t - 4 \tmi = t -24 \log n / \mu$, and introducing the notation $a = t- \tau$, we obtain:
\begin{equation}\label{eq:rwflow}
\infnorm{ \loadv_{t+1} - {\bf \bar\load}} \leq \frac{1}{4} + r + (\delta\dega + r)\bigg(2+ \sum_{a =0}^{24 \log n / \mu}\max_{{ w}\in V} \sum_{v\in V} \left| P_{a+1}(v,w)-P_a(v,w) \right|\bigg).
\end{equation}
Thus, the right-hand side of the above expression provides an asymptotic upper bound on the discrepancy at time $O(\log(nK)/\mu)$. It remains to provide an estimate of the sums which appear in the expression. These sums can be analyzed using techniques for bounding probability change (current) of a reversible random walk in successive time steps. In Appendix~\ref{sec:appendix23}, we provide three different ways of bounding the expression, leading directly to claims (i), (ii), and (iii) of the theorem.

We remark that it is not clear whether the obtained bounds on the right hand side of~\eqref{eq:rwflow} are asymptotically tight. For example, the question whether it may be possible to replace ``$\sqrt n$'' in claim (ii) by a term which is polylogarithmic in $n$ is an interesting open question in the theory of random walks on graphs (cf.~\cite{NPZ15} for some recent related results in the area).
\end{proof}

%
%
%%

%
%
%

%
%
% end merge

%

%We recall the statement of the Theorem~\ref{the:prediscrepancyinreggraphs}.
%
%\medskip\noindent\textbf{Theorem~\ref{the:prediscrepancyinreggraphs}}
%Let $A$ be any cumulatively fair balancer. For any $d$-regular input graph $G$, after $O\left(\frac{\log(K n)}{\mu}\right)$ time, $A$ achieves:
%\begin{enumerate}[(i)]
%\item $O\left( \dego \cdot\sqrt{\frac{\log n}{\mu}} \right)$-discrepancy for $\dega \ge 2\dego$.
%\item $O\left(\dego\cdot \sqrt{n}\right)$-discrepancy  for $\dega\ge 2\dego$.
%\item $O\left( \dego \cdot{\frac{\log n}{\mu}} \right)$-discrepancy for arbitrary $\dega\ge \dego + 1$.
%\end{enumerate}
%where $d^+$ is the degree of the balancing graph (including self-loops) used by $A$.
%\medskip

%All bounds proven in this theorem are obtained by bounding (\ref{eqn:theT}).

%%%%%%%%%%%%%%%%%%%%%%%%%%%%%%%%%%%%%%%%%%%%%%%%%%%%%%%%%%%%%%%%%%%%%%%%%%%%%%%%%%

\section{Results for Good $s$-Balancers}\label{sec:goods}
In this section, we consider a subclass of cumulatively $1$-fair balancers that achieve a better discrepancy compared to the cumulatively fair balancers if the runtime is slightly larger than $T$.
Algorithms of this class are, by definition, a subclass of cumulatively $1$-fair balancers.
 Hence, Theorem~\ref{the:prediscrepancyinreggraphs} also applies to Good $s$-Balancing Algorithms.
A cumulatively $1$-fair balancer is also a \class if the algorithm is \emph{(i)} round-fair and \emph{(ii)} \emph{self-preferring}, \ie, if it favors self-loop edges over original edges.
As we will see in Theorem \ref{the:timecprime}, the ``more self-preferring'' an algorithm is, the faster it balances.
%A cumulatively $\delta$-fair balancer is called

\begin{definition}
Assume $\delta$ is an arbitrary constant and $1\leq s\le \degs$.
An algorithm is called a {\em \class} if
if for
$t\in\natN$ and  $u\in V$ all edges of $u$ (including self-loops)  receive
$\floor*{\load_t(u)/\dega}$  many tokens in step $t$.
The remaining $e(u)=x_t(u)- d^{+}\cdot \floor*{\load_t(u)/\dega}$ have to be distributed over the edges such that

\begin{enumerate}[topsep=1pt, partopsep=0pt]%[(i)]
 \setlength{\itemsep}{1pt}%
    \setlength{\parskip}{0pt}
\item the algorithm is cumulatively $1$-fair
\item at least $\min\{s, e(u)\}$ self-loops receive
$\ceil*{\load_t(u)/\dega}$ many tokens. ($s$-self-preferring)
\item every edge  $\original{e}\in E_u^+$ receives  at most  $\ceil*{x_t(u)/\dega}$ many tokens.
\end{enumerate}

\end{definition}

Note that, by definition, every \class is round-fair, meaning
that every edge receives either $\floor*{x_t(u)/\dega}$ or $\ceil*{x_t(u)/\dega}$ many tokens in every round.

\begin{observation}
The algorithm ${\bf \algoRF}\left([\sfrac{x}{\dega}]\right)$
is a \classpara{(\dega -2d)} for $\dega > 2d$.
Furthermore, \RRstar is a \classpara{1}.
\end{observation}

The following theorem shows that \classs achieve a smaller discrepancy of $O(d)$ if they are allowed to run longer than $T$ steps.
%Note that due to Proposition~\ref{pro:SimulateCumFairness},

\begin{theorem}\label{the:timecprime}
Let $G$ be any $d$-regular input graph and let $d^+$ is the degree of the balancing graph $G^{+}$.
Let $\delta$ be an arbitrary constant and let $1\leq s \leq (d^+-d)$. Assume
$A$ is a \class.
Then $A$ achieves a discrepancy of
$((2\delta+1)\dega+4\degs)$ after time $\bigo\left(\log K +  \dego s^{-1}\cdot \log^2n/\mu \right).$
%Let $z$ be the number of tokens which every interval looses whenever the node goes down to load of less than $c(\gamma/d-i)$. (crossing barrier)
\end{theorem}

We note that large values of $s$ ($s = \Omega(d)$) increase the speed of the balancing process.
%Unfortunately, we are not able to argue that $s$-self-preference is a necessary condition to show the results of  Theorem \ref{the:timecprime}.
In Theorem \ref{the:lowerbound} (Section~\ref{sec:lowerbounds}) we provide a lower
bound of $\Omega(d)$  on the discrepancy of any stateless algorithms, the bound is independent of the balancing time. Since the class of \classs contains many stateless algorithms, this also means that the bound on the discrepancy in Theorem~\ref{the:timecprime} cannot be improved without further restrictions on the class.
%Obtaining better local deterministic balancing schemes, i.e., achieving $o(d)$ discrepancy after a short time, appears to be infeasible.

%%%%%%%%%%%%%%%%%%%%%%%%%%%%%%%%%%%%%%%%%%%%%%%%%%%%%%%%%%%%%%%%%%%%%%%%%%%%%%%%%%

%%%%%%%%%%%%%%%%%%%%%%%%%%%%%%%%%%%%%%%%%%%%%%%%%%%%%%%%%%%%%%%%%%%%%%%%%%%%%%%%%%

%%%%%%%%%%%%%%%%%%%%%%%%%%%%%%%%%%%%%%%%%%%%%%%%%%%%%%%%%%%%%%%%%%%%%%%%%%%%%%%%%%

%\subsection{Analysis of the Class}
%\label{sec:strictly}
The remainder of this section is devoted to a proof of Theorem \ref{the:timecprime}. We first define the following two families of potential functions,  parameterized by $c$:\\
\vspace{-1cm}
\begin{center}

\begin{tabular}{ccc}
\parbox[t]{5.8 cm}{
$$\phi_t(c)= \sum_{v\in V}{  \max\{ \load_t(v) -c\dega,0 \} } \text { and}$$ }		 &
		\parbox[t]{4cm}{
$$\phi'_t(c)= \sum_{v\in V}{  \max\{  c\dega+s-\load_t(v),0 \} }.$$
}
\end{tabular}
\end{center}

\vspace{-0.5cm}
To show the theorem we use Equation \ref{eqn:theT} of the proof of Theorem \ref{the:prediscrepancyinreggraphs}, %giving us a bound on the difference between
%the average load of a node in a time interval of length $\hatt$
%and the average load $\bar x$. The result also
to derive Lemma \ref{cor:dropbelowtheline}. The lemma
shows that, for every node $u$, there exists a time step $t_u$ in which the load of the node has a certain distance to
$\bar x$. We will then show that the time step $t_u$ results in a potential drop of $\phi_t(c)$ for $u$ if the load of $u$ was larger than $c\dega$.

The following lemma gives a bound on the required length of the time interval so that there is a step $t_u$ where $u$
has a load which is sufficiently close to $\bar x$. The required time is expressed as a fraction of $\log n/ \mu$. The lemma shows a tradeoff (parameter $\lambda$) between the required time and the load difference of $u$ to $\bar x$.
The proof can be found in Appendix~\ref{sec:missingstrict}.

\begin{lemma}\label{cor:dropbelowtheline}
Consider any cumulatively $\delta$-fair balancer with remainder bounded by $r$, and an initialization of the load balancing process with average load $\bar x$ and initial discrepancy $K$. Let $\lambda \geq 0$, and  let $t \geq 16 \cdot\log(nK)/\mu$, and let $\hatt= O\left(d\log n /(\mu\cdot (\lambda+1))\right).$ Then we have: \\$\text{For all } u\in V\text{ there exists a time step } t' \in [t+1; t+\hatt]\text{ such that }\load_{t'}(u) \leq  \bar\load + \delta\dega+2r+1/2 + \lambda.$
\end{lemma}
%We note that a bound symmetric to the above Lemma can be established for approaching the average load $\bar x$ from below with a cumulatively fair process.
%\subsection{Potential Drop}\label{s:combinatoricalargs}

The next lemma bounds the one-step potential drop of $\phi_t(c)$ occurring on every node which has a load of more than $c\dega$ at time $t-1$ and has a smaller load of at most $c\dega +s$ at time $t$. The proof can be found in Appendix~\ref{sec:missingstrict}.
%In the proof we suppose that every node which has more than $c\dega$ tokens colors these 'surplus' tokens red. All remaining tokens in the system are colored black.
%The potential $\phi_t(c)$ counts essential the red tokens in the system at time $t$.
%Now if every node distributes first its black tokens evenly over its original edges and afterwards its red tokens, then
%we can see that a red token can never become black.
%However, a black can become red as the following shows:

\begin{lemma}[Monotonicity of potential]\label{potential}
%Let $1\leq s \leq \dega$.
Let $A$ be a \class.
The potential $\phi_t(c)$
is non-increasing in time
and it satisfies: $\phi_t(c) \leq \phi_{t-1}(c) - \sum_{u\in V} \Delta_t(c,u)$, where:
$$
 \Delta_t(c,u)= \left\{\begin{array}{ll}
	 \min\{\load_{t-1}(u) , c\dega + s\}-\max\{\load_{t}(u), c\dega\} & \text{if } x_{t-1}(u) > x_{t}(u)\text{ and } x_{t-1}(u)  > c\dega \\
	 & \phantom{o..}\text{and } x_{t}(u)  < c\dega+s  \\
         0 & otherwise\end{array}\right.
$$
\end{lemma}
%Before proving the lemma, we remark that in clause (\ref{potential:enum2}), the potential drops for any algorithm which is at least $1$-self-preferring, at time $t$ (i.e., $\Delta_t(c,u)\geq 1$) for any node $u$ such that $\load_{t-1}(u) \geq  c\dega + 1$ and $\load_{t}(u) \leq  c\dega$.

The following observation extends Lemma \ref{potential} to intervals $[t,t']$. It estimates
the potential drop of $\phi_t(c)$ for nodes which have a load
$\ge c \dega$ at time $t$ and a load
$\le c\dega$ during one time step of the interval. The observation follows directly from Lemma~\ref{potential}; we omit its proof.

\begin{observation}\label{obs:fastdrop}
Let $A$ be a \class and let $t\le t'$ be two fixed time steps. Denote by $U$ the subset of nodes  such that for all $u\in U$  $\load_t(u) \ge c\dega +1$ and there exists a moment of time $t_u\in[t, t']$ such that $\load_{t_u}(u)  \le c\dega$. Then $ \phi_{t'}(c) \leq \phi_t(c)-\sum_{u\in U} \max \{  s, \load_t(u) - c\dega \}. $
\end{observation}

The potential defined in Lemma \ref{potential} bounds the number of tokens above certain thresholds. Now we use $\phi'_t(c)$  to show symmetric results measuring the number of `gaps' below certain thresholds. The proof of Lemma~\ref{potentialprime} is very similar to that of Lemma~\ref{potential}, and we provide it for completeness in Appendix~\ref{sec:missingstrict}.

\begin{lemma}\label{potentialprime}
Let $A$ be a \class.
The potential $\phi'_t(c)$
is non-increasing in time and it
 satisfies: $\phi'_t(c) \leq \phi'_{t-1}(c) - \sum_{u\in V} \Delta'_t(c,u)$, where:
$$
 \Delta'_t(c,u)=\left\{\begin{array}{ll}
	 \min\{\load_{t}(u) , c\dega + s\}-\max\{\load_{t-1}(u), c\dega\} & \text{if } x_{t-1}(u) < x_{t}(u)\text{ and } x_{t-1}(u)  < c\dega + s\\
	 & \phantom{o..}\text{and } x_{t}(u)  > c\dega  \\
         0 & otherwise\end{array}\right.
$$
%{\color{gray}
%$$
%\Delta_t(c,u) = \max\{\min\{\load_{t-1}(u)-c\dega, s\}-\max\{\load_{t}(u)-c\dega,0\},0\}.
%$$
%}

%\item \label{2enum3} satisfies: $\phi'_t(c) \leq \phi_{t-1}(c) - |\{\{u,v\}\in E: \load_{t}(v) \leq c\dega  \wedge f_{t-1}(u,v) \geq c+1\}|$.

\end{lemma}

Before proving the lemma, we remark that %in clause $(\ref{2enum2})$,
the potential admits a drop at node $u$ at time $t$ (\ie, $\Delta'_t(c,u)\geq 1$) for every node $u$ such that $\load_{t-1}(u) \leq  c\dega$ and $\load_{t}(u) \geq  c\dega+1$, for any algorithm which is at least $1$-self-preferring.
Again, the following observation follows directly from Lemma~\ref{potentialprime}; we omit its proof.

\begin{observation}\label{obs:fastdropprime}
Let $A$ be a \class and let $t\le t'$ be two fixed time steps. Denote by $U$ the subset of nodes  such that for all $u\in U$  $\load_t(u) < c\dega +s$ and there exists a moment of time $t_u\in[t, t']$ such that $\load_{t_u}(u)  \ge c\dega+s$. Then $ \phi'_{t'}(c) \leq \phi_t(c)-\sum_{u\in U} \max \{  s, c\dega + s - \load_t(u)   \}. $
\end{observation}

\vspace{-0.2cm}
\medskip

The main idea of the rest of the proof is the following. We will consider the potential functions $\phi_{t}(c)$
for decreasing values of $c$ and analyze the time $T_c$ it takes to decrease the potentials
 $\phi_{t}(c)$.
The time bound of  Theorem \ref{the:timecprime} is then the sum of the times  $T_c$ for
suitably chosen values of $c$. A symmetrical argument can be used to bound
$\phi'_{t}(c)$. Details of the arguments of the proof are provided in Appendix~\ref{sec:missingstrict}.

%The introduced potentials are related to levels of load of nodes, arranged at multiples of $\dega$. Within the considered class of algorithms, whenever all nodes have a load under a threshold, their load will remain under this threshold (see Lemma \ref{potential}(\ref{potential:enum1})). Moreover, one can define a pair of potentials, one of which decreases whenever a node moves from over some threshold to under it, and the other --- whenever a node moves from under a threshold to over it (see Lemma~\ref{potential} and Lemma~\ref{potentialprime}, respectively). Using these potentials is the main idea in this section.
%In particular, we can use Lemma~\ref{cor:dropbelowtheline} to show that nodes having a load over a certain threshold (having a certain distance to $\bar\load$) will, at some point within the next $O(\log (Kn)/\mu )$ rounds, drop to a load below the threshold. This enforces a drop of one of the mentioned potentials. Eventually, all nodes will have load concentrated around the average load $\bar x$ of the system, and the considered \classs will achieve a discrepancy of $O(\dego)$. %The exact time after which this happens follows from the proof of Theorem~\ref{the:timecprime}.

%
%%%%%%%%%

%%%%%%%%%%%%%%%%%%%%%%%%%%%%%%%%%%%%%%%%%%%%%%%%%%%%%%%%%%%%%%%%%%%%%%%%%%%%%%%%%%

\section{Lower Bounds}\label{sec:lowerbounds}
We start by showing that the cumulative fairness bounds we introduce cannot be completely discarded when improving upon the discrepancy gaps from~\cite{RSW98}. %Namely, we show that there exists an algorithm which is not cumulative having a discrepancy of at least $c\cdot d \cdot \diam$ for some constant $c$.
%Recall that an algorithm is {\em round-fair } if the load which any node $u$  sends over its edges is either $\floor*{\frac{\load_t(u)}{d}}$ or $\ceil*{\frac{\load_t(u)}{d}}$ in every round.
Note that a round-fair balancer is not necessarily cumulatively
$\delta$-fair for any constant $\delta$.
In the following we show that there are round-fair balancers which have a discrepancy of at least $\Omega( \diam(G)\cdot d)$. The proofs of this section are deferred to Appendix~\ref{sec:missinglower}.
\begin{theorem}\label{lower1}
Let $G$ be a $d$-regular graph. There exists an initial distribution of tokens and a round-fair balancer $A$,
such that $A$ cannot achieve a discrepancy better than $(c\cdot \diam(G)\cdot d)$, for some positive constant $c>0$.
\end{theorem}
The following bound shows that the stateless algorithms we design are asymptotically the best possible in terms of eventual discrepancy. Namely, any stateless algorithm is not able to achieve a discrepancy better than $cd$, for some constant $c$. This also means that the bound on the discrepancy, presented in Theorem~\ref{the:timecprime}, cannot be improved in general for the class of \classs.
\begin{theorem}\label{the:lowerbound}
Let $A$ be an arbitrary deterministic and stateless algorithm.
For every even $n$, there exists a $d$-regular graph and an initial load distribution such that $A$ cannot achieve discrepancy better than $c d$, for some positive constant $c>0$.
\end{theorem}
Our final lower bounds concern variants of the \RR.
The next theorem shows that for a graph without self-loops (\ie, $G=\new{G}$)
the best possible discrepancy of the \RR is at least $c\cdot d \cdot \varphi'(G)$, where $\varphi'(G)$ is the odd girth of graph $G$, \ie, the length of the shortest odd length cycle over all nodes of $G$. This gives for  an odd-length cycle of $n$ nodes a discrepancy of at least $c\cdot n$ for some constant $c$.
\begin{theorem}\label{lowerbound:rotorrouter}
Let $G$ be any $d$-regular and non-bipartite graph, and let $\new{d}=d$. Then, there exists an initial load distribution and direction of the rotors  such that \RR cannot achieve discrepancy better than $(c \cdot d\varphi(G))$, for some positive constant $c>0$, where $(2\varphi(G)+1)$ is the odd girth of $G$.
\end{theorem}

\section{Conclusion}
We introduced two classes of deterministic load-balancing algorithms: Cumulative $\delta$-fair balancers and \class. The lower bounds show discrepancies of $\Omega(d \cdot \diam)$ for algorithms which are not cumulatively $\delta$-fair or which do not have any self-loops.
However, there are two main questions which we leave unanswered: 1) How many self-loops are necessary to obtain our bounds? 2) Are the restrictions imposed by \classs necessary to obtain our bounds?

\newpage
\pagenumbering{roman}
\section*{APPENDIX} %AAAAAA
\appendix
\section{Omitted Proofs from Section \ref{s:fair}}\label{sec:missingfair}

We start with a technical lemma giving bounds on the behaviour of the error matrix ${\bf \Lambda} _\tau$ for $\tau \in\natN$ with $\tau > \log n / \mu$. We recall that
$  {\bf \Lambda}_t=\randomwalk^t - \randomwalk^\infty $.

\begin{lemma}\label{lem:boundmix}
Let $ ({\mathbf q}_t)$ be a sequence of vectors parameterized by $t$, let $\bar{\mathbf  q}_t=\randomwalk^\infty {\mathbf  q}_t$, and let $c\in\natN $ be an arbitrary constant.%Let $\vec{z}$ be a vector satisfying $\norm{\vec{z}}_1=1$
\begin{enumerate}[label=(\roman*)]%[(i)]
\item  %\label{lem:boundmix1}
 For $t \geq c\cdot  4 \frac{\log(n\cdot \max_\tau  \infnorm{{\mathbf q}_\tau-\bar{\mathbf  q}_\tau}) }{\mu}$ we have
\begin{equation*}%\label{eqnmixing1}
\infnorm{{\bf \Lambda}_t {\mathbf q}_t}\leq % \infnorm{{\mathbf q_t}}
2^{-c}.
\end{equation*}
\item%\label{lem:boundmix2}
The following bound holds:
\begin{equation*}%\label{eqnmixing2}
\sum_{t \geq \frac{6c\log n}{\mu}}\infnorm{{\bf \Lambda} _t {\mathbf q}_t} \leq  n^{-c} \cdot \max_{\tau\geq  \frac{6c\log n}{\mu}}\{\infnorm{{\mathbf q}_\tau} \}
\end{equation*}

\end{enumerate}
\end{lemma}
\begin{proof}
The proof proceeds by standard arguments (cf. e.g. \cite{LevinPeresWilmer2006}, Chapter 4); we do not attempt to optimize the constants in the claims. Consider the eigendecomposition of matrix $\bf P = X L X^{-1}$, where $L$ is the normalized diagonal matrix of eigenvalues $\text{diag}(1,\lambda_2,\ldots,\lambda_n)$, with $1 - \mu \geq |\lambda_2| \geq \ldots \geq |\lambda_n|$, for $2\leq i \leq n$. We have $\bf \Lambda_t = X L'_t X^{-1}$, where $L'_t = \text{diag}(0,\lambda_2^t,\ldots,\lambda_n^t)$. To show claim (i), we take into account that $\bar q_t$ is an eigenvalue of $\bf P$, so ${\bf \Lambda}_t \bar q_t = 0$, and we can write:
\begin{align*}
\infnorm{{\bf \Lambda}_t {\mathbf q}_t} &= \infnorm{{\bf \Lambda}_t ({\mathbf q}_t - \bf \bar q_t)} \leq \infnorm X \infnorm {L'_t} \infnorm {X^{-1}} \infnorm{{\mathbf q}_t - \bf \bar q_t} \leq n^2(1-\mu)^t \infnorm{{\mathbf q}_t - \bf \bar q_t} \leq\\ &< 2^{-\mu t} n^2 \infnorm{{\mathbf q}_t - \bf \bar q_t}.
\end{align*}
Claim (i) follows directly. To show Claim (ii), observe that we have:
\begin{align*}
\sum_{t = \frac{6c\log n}{\mu}}^{+\infty}\infnorm{{\bf \Lambda}_t {\mathbf q}_t} &\leq \sum_{i=0}^{+\infty}\sum_{t = \frac{6(c+i)\log n}{\mu}}^{\frac{6(c+i+1)\log n}{\mu} - 1} \infnorm{{\bf \Lambda}_t {\mathbf q}_t} \leq \sum_{i=0}^{+\infty}\left(\frac{6c\log n}{\mu} \cdot n^{-6(c+i)+2}\right) \max_{\tau\geq  \frac{6c\log n}{\mu}}\{\infnorm{{\mathbf q}_\tau} \}\leq\\
&\leq 2\cdot \frac{6c\log n}{\mu} \cdot n^{-6c+2} \max_{\tau\geq  \frac{6c\log n}{\mu}}\{\infnorm{{\mathbf q}_\tau} \} < n^{-c}\max_{\tau\geq  \frac{6c\log n}{\mu}}\{\infnorm{{\mathbf q}_\tau} \}.
\end{align*}
\end{proof}

The following proposition shows that
the change in the model (the way the tokens are retained) in the preliminaries of Section \ref{s:fair} in comparison to the original model described in Section \ref{sec:mod} does change the load
sent over any original edge in the graph. In particular,
 every cumulatively fair balancer can be transformed (by shifting tokens from self-loops to the remainder) into an algorithm guaranteeing cumulative fairness on all edges and that this transformed algorithm sends exactly the same load over original edges in every round.

\begin{proposition}\label{pro:SimulateCumFairness}
For any cumulatively $\delta$-fair load-balancing algorithm $A$,
there exists an Algorithm $A'$ with remainder $r \leq \dega$ such that for any $t$ and for all
$(u,v) \in {E}_u^+$
\begin{enumerate}
\item
$
 \left| \cumflow_t(u,v) - \frac{\out_t(u)}{\dega} \right|\leq \delta.$
\item the load sent over $(u,v)$ is the same in $A$ and $A'$.
\end{enumerate}
% $f^A_t(u,v)=f^{A'}_{t}(u,v)$.
\end{proposition}

\begin{proof}
The reformulation of algorithm $A$ as algorithm $A'$ proceeds as follows.
 For all edges $e \in E(G)$ (\ie, except for self-loops), in every step $A'$ places the same amount of load on $e$ as $A$. However, in $A'$ load may be retained on nodes in a different way, being placed in the remainder $r_t(u)$ rather than on self-loops at node $u\in V$. To prove, that it is always possible, we proceed by induction.

Specifically, at a fixed moment of time $t$, let $f_t(e)$ be the amount of load put on an edge $e$ by algorithm $A$, and $f'_t(e)$ be the amount of load put on an edge by $A'$, and let $F_t(e)$ and $F'_t(e)$ be the respective cumulative loads for algorithms $A$ and $A'$. Algorithm $A'$ processes all edges (including self-loops) sequentially and verifies if sending this amount of load along $e$ would satisfy the cumulative fairness condition up to time $t$ with respect to all edges original from $u$ already processed.

Let $e_1$ be the edge or self-loop that violates the cumulative load property for $A'$, that is there exists an incident edge or self-loop $e_2$ such that $|(F'_{t-1}(e_1)+f_t(e_1)) - (F'_{t-1}(e_2)+f_t(e_2))| > \delta$. Since $|f_t(e_1) - f_t(e_2)| \le 1$, and $|F'_{t-1}(e_1)-F'_{t-1}(e_2)| \le \delta$ (from inductive assumption), we get that
\begin{equation}\label{eq:cumf} (F'_{t-1}(e_1)+f_t(e_1)) - (F'_{t-1}(e_2)+f_t(e_2)) \in \{\delta+1,-\delta-1\}\end{equation}
(without loss of generality we can assume that this value is $\delta+1$). Moreover, we can show that for every $e'_2$ such that the pair $e_1,e'_2$ violates cumulative fairness, the value \eqref{eq:cumf} is $\delta+1$ (otherwise $F'_{t-1}(e_1)-F'_{t-1}(e_2) = \delta$ and $F'_{t-1}(e_1)-F'_{t-1}(e'_2) = -\delta$ imply $F'_{t-1}(e'_2)-F'_{t-1}(e_2) = 2\delta$ which contradicts the inductive assumption).
We can also observe that $e_1$ is a loop (attached to vertex $u$), since non-loop edges satisfy cumulatively fairness for $A$.
Thus it is enough to set $f'_t(e_1) = f_t(e_1)-1$ and increase $r_t(u)$ by one (in the mirror scenario with a value of $-\delta-1$ in \eqref{eq:cumf} we would set $f'_t(e_1) = f_t(e_1)+1$ and decrease $r_t(u)$ by one). It is easy to observe that this makes $e_1$ satisfy $\delta$-fairness with every other edge incident to $u$. After processing all edges and self-loops and edges in this way and since every edge receives $\floor*{x_t(u)/\dega}$ tokens, we eventually obtain that cumulative fairness is preserved, and moreover $|r'_t(u)| \leq \dega$.
\end{proof}

\subsection{Proof of Theorem \ref{the:prediscrepancyinreggraphs}}\label{sec:appendix23}
Here we finish the proof of Theorem \ref{the:prediscrepancyinreggraphs} presented in Section \ref{s:fair}.

\begin{proof}
We begin by providing the details of bound \eqref{eqn:boundtequals1}. Let $w$ be the $i$'th node of $V$, then we define $\bf w$ to be the vector, such that ${\bf w}[i]=1$ and $\forall_{j \not= i} {\bf w}[j]=0$. Let $a = t- \tau$. Then, we have:
\begin{align}
 \infnorm{(\randomwalk^{t+1-\tau} - \randomwalk^{t-\tau}) \eps_\tau }
&=\infnorm{(\randomwalk^{a+1} - \randomwalk^{a})\eps_{t-a} }
=\max_{w\in V}\left|  {\mathbf w}^{\top} \left(\randomwalk^{a+1} - \randomwalk^{a}\right)\eps_{t-a}  \right|\notag\\
&=  \max_{w\in V} \left| {\mathbf w}^{\top} \left( \randomwalk^{a+1} - \randomwalk^{a}\right) \left(\sum_{v\in V}\eps_{t-a}(v)\mathbf{ v }\right) \right| \notag\\
%&=  \max_{w\in V} \left| \sum_{v\in V}\eps_{t-a}(v) \left( {\mathbf w} ^{\top} \left( \randomwalk^{a+1} - \randomwalk^{a}\right) \mathbf{ v }\right) \right|&\notag\\
&\leq  \infnorm{\eps_{t-a}} \cdot \max_{w\in V} \sum_{v\in V} \left|   {\mathbf w} ^{\top} ( \randomwalk^{a+1} - \randomwalk^{a}) \mathbf{ v } \right|\notag \\
&\leq (\bounderrorproof)  \cdot \max_{{ w}\in V} \sum_{v\in V} \left| P_{a+1}(v,w)-P_a(v,w) \right|  & \eqref{eqn:boundtequals1}\notag
\end{align}
Since the graph is regular, we have $P_{t}(v,w)=P_{t}(w,v)$.\footnote{For general graphs, one can use $P_{t}(v,w)=(\dega(w)/\dega(v)) P_{t}(w,v)$}
From here on, we split the analysis of the claims of Theorem  \ref{the:prediscrepancyinreggraphs}, proving each one of them by
bounding (\ref{eqn:boundtequals1}) separately.

% (i), (ii) and (iii) separately.
\begin{enumerate}[label=(\roman*)]%[(i)]
\item {\bf $\bigo\left((\delta+1)d\sqrt{\frac{\log n}{\mu}}\right)$-discrepancy for $\dega \geq 2d$ :}\\
For regular graphs having $P(u,u)\geq 1/2$ for all $u\in V$,  we have by %\cite{MARKOV-Mixing 5.11}
\cite{LevinPeresWilmer2006} (for $a > 0$)
$$\forall_{w\in V} \sum_{v\in V}\left|P_{a+1}(w,v)-P_{a}(w,v) \right| < \frac{24}{\sqrt a}. $$
(For case of $a=0$, we have $\forall_{w\in V} \sum_{v\in V}\left|P_{1}(w,v)-P_{0}(w,v) \right| \le 2.$)

We obtain by applying $\sum_{a=1}^{t-t^*}1/\sqrt{a} \leq 2\sqrt{t-t^*} $
$$\sum_{0 \le a< t-t^*}\infnorm{(\randomwalk^{a+1} - \randomwalk^{a})\eps_{t-a} }
< (\bounderrorproof) \left(2+48\sqrt{ \underbrace{t-t^*}_{4\tmi}}\right)
 \le 98(\bounderrorproof)\sqrt{\tmi}.$$
Introducing this into (\ref{eqn:theT}) and setting $\hatt=1$ yields
$$\infnorm{\loadv_{t+1} - \bar\load}=
%\infnorm{(\outv{t+1}-\outv{t}+\xv_{t+1}) - \bar\load} =
\bigo\left((\delta\dega+r)\sqrt{\tmi} \right) = \bigo \left((\bounderror) \sqrt{\frac{\log n}{\mu}} \right)$$ where we take into account that $\tmi = 6\log n/\mu$, and that $r\leq \dega$.
(Observe that whenever a cumulatively fair balancer has a $\delta=0$ the bound on the remainder $r$ has to be of order $\Omega(d)$.)
%%%%%%%%
\item {\bf $\bigo\left((\delta+1)d\sqrt{n}\right)$-discrepancy for $\dega \geq 2d$ :}\\
Let $D_{a+1}$ be the diagonal matrix of $(P_{a+1}-P_{a})$ and let $X$ be the corresponding base change matrix.
\begin{align}
\max_{w \in V} \sum_{v\in V}\left|P_{a+1}(w,v)-P_{a}(w,v) \right|
&= \max_{w \in V} \sum_{v\in V}|{\bf v} (P_{a+1}-P_{a}) {\bf w} | \notag\\
&=\max_{w \in V}  \onenorm{(P_{a+1}-P_{a}) {\bf w}}  \notag\\
&\leq \max_{\twonorm{\bf w}=1}\onenorm{(P_{a+1}-P_{a}) {\bf w}}\notag\\
&= \max_{\twonorm{\bf w}=1}\onenorm{X^{\top} D_{a+1} X {\bf w}}\notag\\
&\leq \sqrt{n}\twonorm{X^{\top}D_{a+1} X }\notag\\
&= \sqrt{n}\twonorm{D_{a+1} X } \label{eqn:boundforsqrtn}.
\end{align}

We have
$$D_{a+1}=\left(
\begin{matrix}
 \lambda_1^{a+1}-\lambda_1^{a} &  0  & \ldots & 0\\
0  &  \lambda_2^{a+1}-\lambda_2^{a} & \ldots & 0\\
\vdots & \vdots & \ddots & \vdots\\
0  &   0       &\ldots &  \lambda_n^{a+1}-\lambda_n^{a}
\end{matrix}
\right).
$$
where $\lambda_1 \kdots \lambda_n$  are the eigenvalues of $P$. We note that $\lambda_1=1$ and $\lambda_2 \kdots \lambda_n\in[0,1]$ since $\degs \geq \dego$.
Hence by plugging (\ref{eqn:boundforsqrtn}) in  (\ref{eqn:boundtequals1}) we derive
\begin{align*}
\sum_{0 \le a< t-t^*}\infnorm{(\randomwalk^{t+\hatt-\tau} - \randomwalk^{t-\tau}) \eps_\tau }
&\leq (\delta+1)\dega \cdot \sum_{a< t-t^*}\max_{ w\in V} \sum_{v\in V}|P_{a+1}(w,v)-P_{a}(w,v) | \\
&\leq (\delta+1)\dega \cdot \sum_{a< t-t^*}\sqrt{n} \twonorm{D_{a+1} X}\\
&\leq (\delta+1)\dega\cdot\sqrt{n}\sum_{a=0}^{t-t^*} |\lambda_2^{a+1}-\lambda_2^{a}  | \\
&=(\delta+1)\dega\cdot\sqrt{n}\cdot(\lambda_2^{0}-\lambda_2^{t-t^*})\\
&\leq(\delta+1)\dega\cdot\sqrt{n}
\end{align*}
Introducing this into  (\ref{eqn:theT}) and setting $\hatt=1$ yields:
$$\infnorm{\loadv_{t+1} - \bar\load}%=\infnorm{(\outv{t+1}-\outv{t}+\xv_{t+1}) - \bar\load}
 \leq (\delta+1)\dega  \cdot \sqrt{n},$$
which completes the proof.
\item {\bf $\bigo\left((\delta+1)\frac{d\log n}{\mu}\right)$-discrepancy for $\dega \geq d+1$:}\\
We bound the term $\sum_{t^* < \tau \leq t}\infnorm{(\randomwalk^{t+\hatt-\tau} - \randomwalk^{t-\tau})\eps_\tau } $ of  (\ref{eqn:theT}):\\

\begin{align*}
\sum_{t^* < \tau \leq t}\infnorm{(\randomwalk^{t+\hatt-\tau} - \randomwalk^{t-\tau})\eps_\tau }
&\leq \sum_{t^* < \tau \leq t}\left(\infnorm{(\randomwalk^{t+\hatt-\tau}\eps_\tau } + \infnorm{\randomwalk^{t-\tau}\eps_\tau }\right) \\
&\leq 2 \sum_{t^* < \tau \leq t}\infnorm{\eps_\tau }\\
&\leq 2(\bounderrorproof)(t-t^*) = 8\tmi(\bounderrorproof).
\end{align*}
Putting this into (\ref{eqn:theT}) and dividing by $\hatt$ gives:
\begin{equation}\label{eqn:lasteqnofkeylemma}
 \infnorm{ \frac{\sum_{t < \tau \leq t+\hatt }\loadv_\tau}{\hatt} -  \bar\load}  \leq \left(\delta\dega+2r+\frac 14\right)+\frac{(8\tmi+1)(\bounderrorproof)}{\hatt}.
\end{equation}

The claim follows from (\ref{eqn:lasteqnofkeylemma}) by setting $\hatt=1$.
\end{enumerate}
\end{proof}

\section{Omitted Proofs from Section \ref{sec:goods}}\label{sec:missingstrict}

\subsection{Proof of Lemma \ref{cor:dropbelowtheline}}\label{lemma5}
\begin{proof}
We build upon the proof of Theorem \ref{the:prediscrepancyinreggraphs} (see Section \ref{s:fair}).
We will use the notation established in the proof of  Theorem \ref{the:prediscrepancyinreggraphs}.
Since $\degs=O(\dega)$, $\delta=O(1)$, and $r\leq \dega$, there is a constant $c$ such that $c\dego \geq \delta\dega+r$.
We also have $\tmi = 6\log n / \mu$.
We set $\hatt \ge 216 c \cdot \frac{d\log n}{\mu (\lambda+1)}$. \\
We derive from (\ref{eqn:lasteqnofkeylemma})

$$ \infnorm{ \frac{\sum_{t < \tau \leq t+\hatt }\loadv_\tau}{\hatt} -  \bar\load}  \leq \delta\dega+2r+\frac 14+\frac{9\tmi \cdot cd}{216c \cdot \frac{d\log(n)}{\mu (\lambda+1)}}
\leq \delta\dega+2r+\frac 14+\frac{(\lambda+1)}{4}
\leq \delta\dega+2r+\frac 14 + \lambda + \frac{1}{4}.$$
%where the last inequality is true for $c'' \geq 36 c' c$.
Therefore, we have that the difference between the average load of any node over $\hatt = \frac{6 d\log(n)}{\mu (\lambda+1)}$ steps and  the average load $\bar\load$ is bounded by  $\delta\dega+2r+\frac 12+ \lambda$.
This means that for every node $u$ there has to be a time step $t'\in [t+1,t+\hatt]$  such that  $x_{t'}(u)-\bar\load \leq   \delta\dega+2r+\frac 12 + \lambda$.
This yields the claim.
\end{proof}

\subsection{Proof of Lemma~\ref{potential}}\label{lemma6}
\begin{proof}
Fix $c\in\natN$. At any time $t$, we will divide the set $L$ of $m$ tokens circulating in the system into two groups: the set of \emph{black} tokens $L_t^-$ and the set of \emph{red} tokens $L_t^+$, with  $L = L_t^- \cup L_t^+$. Colors of tokens persist over time unless they are explicitly recolored. For $t=1$, for each node $u$ we color exactly $|L_1^-(u)|= \min\{\load_{1}(u), c\dega\}$ tokens at $u$ black, and the remaining tokens at $u$ red. In every time step, we follow two rules concerning token distribution:
\renewcommand{\labelenumi}{(\arabic{enumi})}
\begin{enumerate}
\item \label{rule1}The number of black tokens leaving a node $u$ along any edge (including self-loops) is never more than $c$.
\item \label{rule2}At the start of each subsequent step $t$, we recolor some red tokens to black, so that the total number of black tokens located at a node $u$ is exactly $|L_t^-(u)|= \min\{\load_{t}(u), c\dega\}$.
\end{enumerate}
\renewcommand{\labelenumi}{\arabic{enumi}.}
We note that both rules of token circulation are well defined. The proof proceeds by induction. To prove the correctness of rule (1) at step $t$, observe that, by the definition of \classs, for any node $u$ we either have $\load_{t}(u) \leq c\dega$ and then node $u$ sends at most $c$ tokens along each of its edges and self-loops, or $\load_{t}(u) > c\dega$, and then node $u$ sends at least $c$ tokens along each of its edges and self-loops. In the first case, rule (\ref{rule1}) is correct regardless of how $u$ distributes tokens of different colors; in the second case, $u$ has exactly $c\dega$ black tokens, and we can require that it sends exactly $c$ of its black tokens along each of its edges and self-loops. To prove the correctness of rule (\ref{rule2}), we note that by the correctness of rule (\ref{rule1}) for the preceding time step, the number of black tokens arriving at $u$ along edges and self-loops can be upper-bounded by $\min\{\load_{t}(u), c\dega\}$. Hence, no recoloring of tokens from black to red is ever required.

We now observe that the potential $\phi_t(c)$ is by definition the number of red tokens circulating in the system at any given moment of time. Indeed, we have:
$$
\phi_t(c) = \sum_{u\in V} (\load_{t}(u) - \min\{\load_{t}(u), c\dega\}) = m - \sum_{u\in V} |L_t^-(u)| = m - |L_t^-| = |L_t^+|.
$$
The monotonicity %condition (\ref{potential:enum1})
 for the potential follows immediately from the fact that no new red tokens appear in the system. To prove %(\ref{potential:enum2}),
the potential drop, we will show that the number of tokens being recolored from red to black at node $u$ in step $t$ is at least $\Delta_t(c,u)$. Indeed, suppose that at time $t-1$ we had for some node $u$: $\load_{t-1}(u) = c\dega + i$, for some $i\geq 1$. Then, by definition of the self-preference of algorithm $A$, at least $c+1$ units of load will be sent on at least $i' = \min\{i,s\}$ self-loops of $u$ in step $t$. By rule (\ref{rule1}) of the token circulation process, each of these self-loops will contain at least one red token. Thus, the number of red tokens arriving at $u$ at time $t$ is at least $i'$. On the other hand, the number of red tokens remaining after the recoloring at $u$ in step $t$ is precisely $\max\{\load_{t}(u)-c\dega,0\}$. Thus, the number of tokens recolored from red to black at $u$, or equivalently the potential drop induced at $u$,
 is at least
$$
\max\{ i'-\max\{\load_{t}(u)-c\dega,0\},0\}  = \max\{\min\{\load_{t-1}(u)-c\dega, s\}-\max\{\load_{t}(u)-c\dega,0\},0\} \overset{\text{def}}= \Delta_t(c,u)
$$
which yields the claimed potential drop.
%which completes the proof of (\ref{potential:enum2}).
\end{proof}

\subsection{Proof of Lemma~\ref{potentialprime}}
\begin{proof}%[]
The proof is similar to Lemma \ref{potential}.
Fix $c\in\natN$. At any time $t$, we will divide the set $L$ of tokens circulating in the system into two groups: the set of \emph{black} tokens $L_t^-$ and the set of \emph{red} tokens $L_t^+$, with  $L = L_t^- \cup L_t^+$. Colors of tokens persist over time unless they are explicitly recolored. For $t=1$, for each node $u$ we color exactly $|L_1^-(u)|= \min\{\load_{1}(u), c\dega+s\}$ tokens at $u$ black, and the remaining tokens at $u$ red. In every time step, we follow two rules concerning token distribution:
\renewcommand{\labelenumi}{(\arabic{enumi})}
\begin{enumerate}
\item\label{app_enum1} The number of black tokens leaving $u$ along original edges is at most $c$.
\item \label{app_enum2} At the start of each subsequent step $t$, we recolor some red tokens to black, so that the total number of black tokens located at a node $u$ is exactly $|L_t^-(u)|= \min\{\load_{t}(u), c\dega +s\}$.
\end{enumerate}
\renewcommand{\labelenumi}{\arabic{enumi}.}
We note that both rules of token circulation are well defined. The proof proceeds by induction. To prove the correctness of rule (\ref{app_enum1}) at step $t$, observe that, by the definition of \classs, for any node $u$ we either have $\load_{t}(u) \leq c\dega$ and then node $u$ sends at most $c$ black tokens along each of its edges and self-loops, or $\load_{t}(u)  > c\dega$, and then node $u$ sends over $\max\{\min \{x_t(u) - c\dega, s  \},0\}$ many self-loops $c+1$ black tokens and $c$ along all other edges. In the first case, rule (\ref{app_enum1}) is correct regardless of how $u$ distributes tokens of different colors; in the second case, let $s'=\max\{\min \{x_t(u) - c\dega, s  \},0\}$.
Node $u$ has exactly $ c\dega+s'$  black tokens, and we can require that it sends exactly $c+1$ of its black tokens along $s'$ many arbitrary self-loops, since the algorithm is $s$-self-preferring, and exactly $c$ along each of its other edges. To prove the correctness of rule (\ref{app_enum2}), we note that by the correctness of rule (\ref{app_enum1}) for the preceding time step, the number of black tokens arriving at $u$ along edges and self-loops can be upper-bounded by $\min\{\load_{t}(u), c\dega +s\}$. Hence, no recoloring of tokens from black to red is ever required.

We now observe that the potential $\phi'_t(c)$ is by definition the number of missing black tokens such that every node has $c\dega+s$ of them. Indeed, we have:
$$
\phi'_t(c) = \sum_{u\in V} (c\dega\ +s - \min\{\load_{t}(u), c\dega + s\}) = (c\dega + s)\cdot n - \sum_{u\in V} |L_t^-(u)| = (c\dega + s)\cdot n - |L_t^-|.
$$
The monotonicity %condition (\ref{2enum1})
for the potential follows immediately from the fact that no new red tokens appear in the system. To prove %(\ref{2enum2}),
the claimed potential drop, we will show that the number of tokens being recolored from red to black in time step $t$ is at least $\Delta'_t(c,u)$. Note, that a red token, which is recolored in black, will decrease the potential by 1.\\
Indeed, suppose that at time $t-1$ we had for a node $u$: $\load_{t-1}(u) = c\dega +s - i$, for an integer $i\geq 1$.
%i is the number of missing tokens
Then, by the definition of the self-preference of algorithm $A$, at
least $i'=\min\{i,s \}$ self-loops carry at most $c$ tokens in step $t$. Intuitively, each of them can 'trap' a red token.
 By rule (\ref{app_enum1}) of the token
circulation process, every neighbor of $u$ sent at most $c$ black tokens.
Thus, the number of black tokens arriving at $u$ at time $t$ is at most $c\dega + s - i'$, and the number of red tokens which $u$ receives is at least $\max\{x_t(u) - (c\dega + s-i'),0\}$.
Therefore, for $\load_{t-1}(u) < c\dega + s$, since at most $i'$ self-loops 'trap' a red token we have that the number of red tokens which are repainted black at node $u$ at time $t$ is at least (by rule (\ref{app_enum2}) of recoloring)
\begin{align*}
&\min\{ \max\{x_t(u) - (c\dega + s-i'),0\},i'\}
\\&=
\min  \{ \max\{\
\load_{t}(u)-(c\dega -\min\{c\dega -\load_{t-1}(u),0 \}),0\},\min\{  c\dega +s -\load_{t-1}(u)  ,s \}  \}
\\&=  \max\{\min\{x_t(u) - x_{t-1}(u),s,x_t(u) -c\dega,c\dega +s -\load_{t-1}(u)\},0\} = \Delta'_t(c,u)
,\end{align*}
and for $\load_{t-1}(u)\geq c\dega+s$ we have $\Delta'_t(c,u)= 0$,
which yields the claimed potential drop.
\end{proof}

\subsection{Proof of Theorem \ref{the:timecprime}}

\begin{proof}
We consider the behavior of the process for $t\ge T=\bigo(\frac{\log(Kn)}{\mu})$. Due to the monotonicity described in Lemma \ref{potential},
once the maximum load in the system is below $c\dega$ for some c, it will stay below this threshold. In particular, this means that the maximum load will not exceed $x_m = \bar x + \bigo\left(\frac{d\log n}{\mu}\right)$ after $\bigo\left((\delta+1) \dego \cdot{\sfrac{\log n}{\mu}} \right)$ time, as shown
in Theorem~\ref{the:prediscrepancyinreggraphs}\ref{en1:part3}.

In the first part of the proof, we will show that after a further $\bigo(\frac{\dega}{s}\frac{\log^2 n }{\mu })$ time steps, the maximum load in the network will drop below the threshold value $c_0 d^+$, where $c_0$ is the smallest integer such that $c_0 d^+ \geq \bar\load +\delta d^+ + 2\degs + \dega/2.$

We note that if there exists a node $u$ with load $x_t(u) \geq c_0\dega+1$ at some time moment $t\geq T$, then by Lemma~\ref{cor:dropbelowtheline} with $\lambda = d^+/2-1/2$ (which we can apply to \classs, putting $r=\degs$, by Proposition~\ref{pro:SimulateCumFairness}), there exists some time step $t' \in [t+1; t + \hatt]$, where $\hatt = \bigo\left( \frac{\log n }{\mu }\right)$, such that we have $x_{t'}(u) \leq c_0\dega$. We then obtain directly from Observation~\ref{obs:fastdrop} that such a load change for node $u$ results in the decrease of potential $\phi(c_0)$ in the time interval $[t+1; t + \hatt]$. Since the potential $\phi(c_0)$ is non-increasing and non-negative, it follows that eventually the load of all nodes must be below the threshold $c_0\dega$.

In order to prove that the load of all nodes drop below $c_0\dega$ within $\bigo(\frac{\dega}{s}\frac{\log^2 n }{\mu })$ time steps after time $T$, we apply a more involved potential-decrease argument based on the parallel drop of multiple potentials $\phi(c)$, for $c\in\{c_0,c_0+1, \ldots, c_1\}$. Here, $c_1$ is the smallest integer such that $c_1 d^+ \geq x_m = \bar\load + \bigo\left(\frac{d\log n}{\mu}\right)$.

%We have
%$$
%\phi_t(i)  = \sum_{j=1}^{i}s_t(c_1 - j).
%$$
We now partition the execution of our process into phases of duration $t_p$, $p = 1,2,3,\ldots, p_f$, such that, at the end of moment $T_p = T + t_1 + \ldots + t_p$ (end of the $p$'th phase), the following condition is satisfied:
\begin{equation}\label{eqn:a1}
\phi_{T_p}(c) \leq 4^{(c_1-c)} 2^{-p} (c_1-c_0)d^+ n, \quad\text{for all $c_0 \le c \le c_1$.}
\end{equation}
As we have remarked, the values of time $t_p$ are well defined, since eventually the potentials $\phi(c)$ drop to $0$, for all $c_0 \leq c \leq c_1$. Our goal is now to bound the ending time $T_{p_f}$ of phase $p_f$, where:
\begin{equation}\label{eqn:a2a}
p_f = 2(c_1-c_0) + \lceil \log ((c_1-c_0)d^+ n)\rceil + 1.
\end{equation}

At the end of this phase, we will have by~\eqref{eqn:a1} that $\phi_{T_{p_f}}(c_0) \leq 1/2$, hence $\phi_{T_{p_f}}(c_0) = 0$, and so there are no nodes having load exceeding $c_0 d^+$.

We now proceed to show the following bounds on the duration of each phase $t_p$:
\begin{equation}\label{eqn:a2}
t_p = \bigo\left(\frac{d^+}{s}\frac{d\log n }{\mu\cdot\max\{(2(c_1 - c_0) -p)\dega+1,\dega/2+1\}}\right).
\end{equation}
For any fixed $p$, assume that bound~\eqref{eqn:a2} holds for all phases before $p$. For phase $p$, the proof proceeds by induction with respect to $c$, in decreasing order of values: $c=c_1, c_1-1 \kdots c_0$.

First, we consider values of $c \geq c_1 - p/2$. For a fixed $c$, following \eqref{eqn:a1} we denote $b = 4^{(c_1-c)} 2^{-p} (c_1-c_0)d^+ n$. Knowing that $\phi_{T_{p-1}}(c) \leq 2b$, $\phi_{T_{p-1}}(c+1) \leq b/2$, and by the inductive assumption $\phi_{T_{p}}(c+1) \leq b/4$, we will show that $\phi_{T_{p}}(c) \leq b$.
Let $s_t(c) := \phi_t(c) - \phi_t(c+1)$; intuitively, $s_t(c)$ can be seen as total the number of tokens in the system which are ``stacked'' on their respective nodes at heights between $c d^++1$ and $(c+1) d^+$. For  $t\in[T_{p-1};T_p]$, $s_t(c)$ satisfies the following bound:

\begin{equation}\label{eqn:a3}
s_t(c) = \phi_t(c) - \phi_t(c+1) \geq \phi_t(c) - \phi_{T_{p-1}}(c+1) \geq \phi_t(c) - \tfrac{b}{2}.
\end{equation}
For any such time moment $t$ consider the set of nodes with load at least $c\dega$ at time $t$.
Within the time interval $[t+1; t+\hatt]$, where the period of time $\hatt = \bigo\left(\frac{d\log n }{\mu\max\{(2(c - c_0)\dega + 1),\dega/2+1\}}\right)$ follows from Lemma~\ref{cor:dropbelowtheline},
every node with a load of more than $c\dega$ at time $t$, will decrease its load below $\bar\load +  \delta\dega+2\degs+\frac 12 +\frac{1}{2}\max\{(2(c - c_0)\dega),\dega/2\} \le c_0 d^+ - \dega/2  + \frac12 + \max\{(\dega(c - c_0)),\dega/4\}  \le c\dega + \frac12$ (and so also below $cd^+$) at some moment of time during the considered time interval.
By Observation \ref{obs:fastdrop} a potential drop occurs for $\phi(c)$ in the considered interval $[t+1; t+\hatt]$. More precisely, every node $u$ with $x_t(u)\in [c\dega, c\dega +s]$ contributes $x_t(u) - c\dega$ to both the potential drop and the value of $s_t(c)$,
whereas every node $u$ with  $x_t(u) >  c\dega +s$  contributes exactly $s$ to the potential drop and at most $d^+$ to the value of $s_t(c)$.
Hence, we obtain from Observation \ref{obs:fastdrop} (and the fact that $s \le d^+$):
\begin{equation}\label{eqn:a4}
\phi_{t+\hatt}(c) \leq \phi_t(c) - \frac{s}{d^+}\cdot s_t(c).
\end{equation}
Combining \eqref{eqn:a3} and \eqref{eqn:a4}, we obtain for any time moment $t\in[T_{p-1},T_p]$:
\begin{equation}\label{eqn:a5}
\phi_{t+\hatt}(c) \leq \phi_t(c) - \frac{s}{d^+}\cdot (\phi_t(c) - \tfrac{b}{2}).
\end{equation}
We can transform this expression to the following form:
$$
\phi_{t+\hatt}(c)-\tfrac{b}{2} \leq \phi_t(c)-\tfrac{b}{2} - \frac{s}{d^+}\cdot (\phi_t(c) - \tfrac{b}{2}) = \left(1-\frac{s}{d^+}\right)\cdot (\phi_t(c) - \tfrac{b}{2}).
$$
Observe that we can fix $t_p$ satisfying \eqref{eqn:a2} so that $t_p \geq  {\frac{d^+}{s}}\hatt$ (which implies $T_p \ge T_{p-1} + {\frac{d^+}{s}}\hatt$) where we took into account that $2(c - c_0)  \geq 2(c_1 - c_0) -p$ for $c \geq c_1 - p/2$. Now, taking advantage of the monotonicity of potentials, we have from \eqref{eqn:a5}:
$$
\phi_{T_{p}} \leq \phi_{(T_{p-1} + {\frac{d^+}{s}}\hatt)}(c) \leq \tfrac{b}{2} + \left(1-\frac{s}{d^+}\right)^{\frac{d^+}{s}}\cdot (\phi_{T_{p-1}}(c) - \tfrac{b}{2}) \leq \tfrac{b}{2} + \tfrac{1}{2} (2b - \tfrac{b}{2}) = \tfrac34b \le b,
$$
which completes the inductive proof of the bound on $t_p$ for $c \geq c_1 - p/2$.

Moreover, for $c < c_1 - p/2$,  \eqref{eqn:a1} holds because $4^{(c_1-c)} 2^{-p} > 1$, and $\phi_{T_p}(c) \leq(c_1-c_0)d^+ n$ holds by the definition of potentials.

Now, taking into account \eqref{eqn:a2a} and \eqref{eqn:a2}, we can bound the time of termination of phase $p_f$ of the process as follows:

\begin{align*}
T_{p_f} &= T + \sum_{p=1}^{p_f} t_{p_f} \notag\\
& =  T + \sum_{p=1}^{p_f} \bigo\left(\frac{d^+}{s}\frac{\dego\log n }{\mu\cdot\max\{(2(c_1 - c_0) -p)\dega+1,\dega/2+1\}}\right)\notag\\
&= \bigo\left(T + \frac{d}{s}\frac{\log n}{\mu}\left(\left(\sum_{p=1}^{2(c_1-c_0)-1} \frac{1}{2(c_1 - c_0) - p+1/\dega}  \right)+ \frac{\dega}{\dega/2+1} \cdot\log \left((c_1-c_0)d^+ n\right)\right)\right)\notag\\
& =
\bigo\left(T + \frac{d}{s}\frac{\log^2 n}{\mu}\right),
\end{align*}
where we recall that $p_f$ was given by expression~\eqref{eqn:a2a}, and that $c_1\dega - c_0\dega = \bigo\left(\frac{d\log n}{\mu}\right)$.

In this way, we have shown that a balancedness of $(\delta+1/2)\dega+2\degs$ is achieved in time $\bigo\left(T + \frac{d}{s}\frac{\log^2 n}{\mu}\right)$. By using the same techniques and Observation~\ref{obs:fastdropprime} instead of Observation~\ref{obs:fastdrop}, we can show that no node has a load of less than $\bar\load - (\delta+1/2)\dega+2\degs$. This gives the desired discrepancy bound of $(2\delta+1)\dega+4\degs$.
\end{proof}

\section{Omitted Proofs from Section \ref{sec:lowerbounds}}\label{sec:missinglower}

\subsection{Proof of Theorem \ref{lower1}}
\begin{proof}
We will describe an initial state, corresponding to a steady state distribution, such that the flow of load along each edge $e$ is the same in every moment of time ($f_0(e) = f_1(e) = f_2(e) = \ldots$), and the load of nodes does not change in time (however, the load of two distant nodes in the graph will be sufficiently far apart).
We take two vertices $u$ and $w$ such that the distance between them is $\diam(G)$.
We assign to every node $v \in V$ a value $b(v)$ being the shortest path distance from $v$ to $u$
($b(u) = 0$, $b(v)=1$ for direct neighbors of $u$, etc.).
For any given edge $(v_1,v_2)$, we assign
$$f_0(v_1,v_2) = \min(b(v_1),b(v_2))$$
We observe, that for each $v$:
$$\max_{e_1,e_2 \in E_{v}} | f_0(e_1)-f_0(e_2) | \le 1$$
and for each edge $(v_1,v_2)$:
$$f_0(v_1,v_2) = f_0(v_2,v_1).$$
Thus, at each step there exists a way to assign values of $\ceil{f(v)}$ and $\floor{f(v)}$ so as to achieve desired values over edges, and that the system is in the steady state.  The sought value of discrepancy is achieved for the considered pair of nodes  $u$ and $w$ whose distance in $G$ is $\diam(G)$.
\end{proof}

\subsection{Proof of Theorem \ref{the:lowerbound}}
\begin{proof}%[Proof of Theorem \ref{the:lowerbound}]
We take an arbitrary graph $G$ with $n$ vertices which contains a $\floor{d/2}$-clique $C$.
We can construct such a graph by taking nodes numbered from $0$ to $n-1$ and connecting each pair of nodes $i$ and $j$ with an edge if and only if $(i-j) \bmod n \in \{n-\floor{d/2},\ldots,n-1,0,1,\ldots,\floor{d/2}\}$. If $d$ is odd, we also add edges $(i,j)$ for all $(i-j) \bmod n = n/2$.
W.l.o.g. we can assume that $C = \{0,1,\ldots,\floor{d/2}-1\}$.

Let $A$ be a deterministic and stateless algorithm. Since $A$ is deterministic and stateless, for each node $u$, at round $t$ the load which $A$ sends to neighbors and the load it keeps through the remainder vector depend solely on the current load of $u$. Let us fix $\ell= |C|-1$.
Initially, the load is distributed in such a way that every node in $C$ has load $\ell$ and every other node has load $0$.
For any given node whose current load is $\ell$, let $p^\circ$ denote the number of tokens kept by $A$ at the considered node,
and let $p_1,p_2,\ldots,p_d$ be the number of tokens sent by $A$ along respective original edges. Clearly, $\ell = p^\circ + \sum_{i=1}^{d} p_i$. At most $\ell$ of those values are positive, so we can assume w.l.o.g. that $p_d=p_{d-1}=\ldots=p_{\ell+1}=0$.

We complete the construction in such a way that at each time step the load over every node is preserved. Let us fix $i \in C$.
We design an adversary which has control over which values from $\{p_1,\ldots,p_d\}$ are sent along edges of the clique,
and which chooses to send along edges of the clique the possibly nonzero values $p_1,p_2,\ldots,p_{\ell}$. These values will be assigned to the edges $(i,(i+1) \bmod d),(i,(i+2) \bmod d),\ldots,(i,(i-1)\bmod d)$, respectively. We assign all other values arbitrarily since they are all equal to 0. Thus, we observe that at each step loads of nodes are preserved, since the new load of all nodes having load $\ell$ at the end of a step is $p^{\circ} + \sum_{i=1}^{\ell} p_i = \ell$ in the next step. All other nodes in $V\setminus C$ will not receive any tokens and they will remain with load 0.
Thus, the load of nodes in the graph does not change over rounds and the load difference of nodes in $C$ and the nodes in $V\setminus C$ is $cd$, for some constant $c>0$.
\end{proof}

\subsection{Proof of Theorem \ref{lowerbound:rotorrouter}}
\begin{proof}
Let $u$ be an arbitrary vertex belonging to the shortest odd cycle.
We assign to every node $v \in V$ a value $b(v)$ being the shortest path distance from $v$ to $u$
($b(u) = 0$, $b(v)=1$ for direct neighbors of $u$, etc.).
Observe that for any edge $(v_1,v_2)$, we have $b(v_1)-b(v_2) \in \{-1,0,1\}$. Moreover, we have $b(v_1) = b(v_2)$ only if $b(v_1) \ge \varphi(G)$. Suppose $b(v_1) < \varphi(G)$. We can construct an odd length cycle $u \leadsto v_1 \to v_2 \leadsto u$ of at most $2\cdot \varphi(G)-1$, a contradiction.

For the construction, fix a sufficiently large integer $L>0$ which will be linked to the average load in the system (it has no effect over the final discrepancy, we need $L$ to be large enough to have nonnegative values of load). We will design a configuration of load in the system which will alternate between two different states, identical for all configurations in odd time steps and even time steps, respectively (thus, $f_0(e) = f_2(e) = \ldots$ and $f_1(e) = f_3(e) = \ldots$). We will describe a configuration at any time step by providing values distributed over every edge. The load distributed over edge $(v_1,v_2)$ will only depend on the values of $b(v_1)$ and $b(v_2)$.
If $b(v_1) \ge \varphi(G)$ or $b(v_2) \ge \varphi(G)$, we set $f_0(v_1,v_2) = L$, otherwise:
$$f_0(v_1,v_2) =
\begin{cases}
L+\left(\varphi(G)-\min(b(v_1),b(v_2))\right)& \text{if } 2 | b(v_1) \text{ and } 2 \!\!\not| \, b(v_2),\\
L-\left(\varphi(G)-\min(b(v_1),b(v_2))\right)& \text{if } 2 \!\!\not| \, b(v_1) \text{ and } 2 | b(v_2).
\end{cases}
$$
We also set:
\begin{equation}
\label{eq:flow_rr}
f_1(v_1,v_2) = f_0(v_2,v_1).
\end{equation}
Setting all of $f_0(e)$ and $f_1(e)$ is enough to describe every value over every edge. We now prove that such a configuration is possible for some execution of the \RR algorithm, \ie, that there exists an ordering of the original edges of each node in the cycle of the \RR which leads to such alternating configurations.
We observe that $f_t(v_1,v_2) + f_t(v_2,v_1) = 2L$ for $t\in\natN$. Thus, for $t\in\natN$ we have
\begin{equation}
\label{eq:symmetry_rr}
f_t(v_1,v_2) + f_{t+1}(v_1,v_2) = 2L.
\end{equation}
We observe, that for any node $v$ and its two neighbors $v_1,v_2$, it holds
$$ |f_t(v,v_1) - f_t(v,v_2)| \le 1$$
Also, due to (\ref{eq:symmetry_rr}):
$$\ldots = f_t(v,v_1) - f_t(v,v_2) = -\left(f_{t+1}(v,v_1) - f_{t+1}(v,v_2)\right) = f_{t+2}(v,v_1) - f_{t+2}(v,v_2) = \ldots$$
By (\ref{eq:flow_rr}), the incoming and original flows of load through edges are preserved, and because the difference between original flows is alternating in signs (for two incident original edges), it is always possible to choose an edge ordering in the cycle of the \RR representing such a situation.
Indeed, we observe that for a particular vertex $v$, original directed edges of $v$ can be partitioned into two sets $P_1 \cup P_2$, where in even steps edges from $P_1$ are given one more token than edges from $P_2$, and in odd steps edges from $P_1$ are given one less token than edges from $P_2$. So it is enough to select an ordering of edges for the \RR such that every edge from $P_1$ precedes every edge from $P_2$ set.

We observe that the node $u$ alternates between loads $(L+\varphi(G))\cdot d$ and $(L-\varphi(G))\cdot d$, while the average load of a node in this setting is exactly $L \cdot d$, which gives us the claimed discrepancy.
\end{proof}


\begin{thebibliography}{}

\end{thebibliography}


\newpage
{
\small
\begin{thebibliography}{10}

\bibitem{AB12b}
C.~Adolphs and P.~Berenbrink.
\newblock {\em Distributed self-prefering load balancing with weights and speeds. }
\newblock In {\em PODC}, pages 135--144, 2012.

\bibitem{AB12}
C.~Adolphs and P.~Berenbrink.
\newblock {\em Improved bounds on diffusion load balancing. }
\newblock In {\em IPDPS}, 2012.

\bibitem{AB13}
H.~Akbari and P.~Berenbrink.
\newblock {\em Parallel rotor walks on finite graphs and applications in discrete
  load balancing.}
\newblock In {\em SPAA}, pages 186--195, 2013.

\bibitem{ABS12}
H.~Akbari, P.~Berenbrink, and T.~Sauerwald.
\newblock {\em A simple approach for adapting continuous load balancing processes to
  discrete settings. }
\newblock In {\em PODC}, pages 271--280, 2012.

\bibitem{BCFFS11}
P.~Berenbrink, C.~Cooper, T.~Friedetzky, T.~Friedrich, and T.~Sauerwald.
\newblock {\em Randomized diffusion for indivisible loads. }
\newblock In {\em SODA}, pages 429--439, 2011.

\bibitem{CDST07}
J.~Cooper, B.~Doerr, T.~Friedrich, and J.~Spencer.
\newblock Deterministic random walks on regular trees.
\newblock {\em Random Struct. Algorithms}, 37(3):353--366, 2010.

\bibitem{CS06}
J.~Cooper and J.~Spencer.
\newblock Simulating a random walk with constant error.
\newblock  {\em Combinatorics, Probability and Computing}, 15:815–822, 2006.

\bibitem{DF09}
B.~Doerr and T.~Friedrich.
\newblock Deterministic random walks on the two-dimensional grid.
\newblock {\em Combinatorics, Probability and Computing}, 18(1-2):123--144,
  2009.

\bibitem{FGS10}
T.~Friedrich, M.~Gairing, and T.~Sauerwald.
\newblock {\em Quasirandom load balancing.}
\newblock In {\em SODA}, pages 1620--1629, 2010.

\bibitem{FS09}
T.~Friedrich and T.~Sauerwald.
\newblock {\em Near-perfect load balancing by randomized rounding. }
\newblock In {\em STOC}, pages 121--130, 2009.

\bibitem{FS10}
T.~Friedrich and T.~Sauerwald.
\newblock The cover time of deterministic random walks.
\newblock In {\em COCOON}, pages 130--139, 2010.

\bibitem{KKM12}
S.~Kijima, K.~Koga, and K.~Makino.
\newblock {\em Deterministic random walks on finite graphs. }
\newblock In {\em ANALCO}, pages 16--25, 2012.

\bibitem{KP14}
A.~Kosowski and D.~Pająk.
\newblock {\em Does Adding More Agents Make a Difference? A Case Study
               of Cover Time for the Rotor-Router. }
\newblock In {\em ICALP}, pages 544--555, 2014.

\bibitem{LevinPeresWilmer2006}
D.~Levin, Y.~Peres, and E.~Wilmer.
\newblock {\em {Markov chains and mixing times}}.
\newblock {\em American Mathematical Society, 2006. }

\bibitem{NPZ15}
J. Norris, Y. Peres, A. Zhai.
\newblock {\em {Surprise probabilities in Markov chains}}.
\newblock In {\em SODA}, pages 1759--1773, 2015.


\bibitem{PDDK96}
V.~Priezzhev, D.~Dhar, A.~Dhar, and S.~Krishnamurthy.
\newblock {\em Eulerian walkers as a model of self-organized criticality. }
\newblock {\em Phys. Rev. Lett.}, 77:5079--5082, 1996.

\bibitem{RSW98}
Y.~Rabani, A.~Sinclair, and R.~Wanka.
\newblock {\em Local divergence of Markov chains and the analysis of iterative
  load-balancing schemes. }
\newblock In {\em FOCS}, pages 694--703, 1998.

\bibitem{SS12}
T.~Sauerwald and H.~Sun.
\newblock{\em Tight bounds for randomized load balancing on arbitrary network
  topologies. }
\newblock In {\em FOCS}, pages 341--350, 2012. A revised and extended version is available online at: http://arxiv.org/abs/1201.2715.

\bibitem{SS94}
R.~Subramanian and I.~Scherson.
\newblock{\em An analysis of diffusive load-balancing. }
\newblock In {\em  SPAA}, pages 220--225, 1994.

\bibitem{journals/corr/ShiragaYKY13}
T.~Shiraga, Y.~Yamauchi, S.~Kijimam, and M.~Yamashita.
\newblock {\em Deterministic Random Walks for Rapidly Mixing Chains.}
\newblock {\tt arXiv:1311.3749}, 2013.

\bibitem{YWB03}
V.~Yanovski,  I.~Wagner,  and A.~ Bruckstein.
\newblock {\em A Distributed Ant Algorithm for Efficiently Patrolling a Network}
\newblock  {\em Algorithmica}, 37(3):165-186, 2003.


\end{thebibliography}

}
\end{document}